\documentclass[11pt]{article}
\usepackage{latexsym,amssymb,amsmath,amsthm,enumerate,geometry,float,cite,tikz}
\newtheorem{theorem}{Theorem}
\newtheorem{lemma}[theorem]{Lemma}

\newtheorem{corollary}[theorem]{Corollary}
\newtheorem{definition}[theorem]{Definition}

\usepackage{lineno}
\usepackage{setspace}
\usepackage{mathtools}
\usepackage{enumitem}
\usepackage{pgfplots}

\renewcommand{\P}{P}
\newcommand{\NP}{NP}

\newcommand{\Hness}{hardness}
\newcommand{\NPH}{NP-hard}
\newcommand{\NPHness}{NP-\Hness}

\newcommand{\NPcness}{NP-completeness}
\newcommand{\NPc}{NP-complete}
\newcommand{\coNPc}{co-NP-complete}
\newcommand{\coNPH}{co-NP-hard}
\newcommand{\coNP}{co-NP}
\newcommand{\coNPHness}{co-NP-completeness}

\newcommand{\APXc}{APX-Complete}
\newcommand{\YES}{YES}
\newcommand{\NOi}{NO}

\newcommand{\pname}[1]{\textsc{#1}}

\usepackage{mathrsfs,dsfont}

\usetikzlibrary{decorations.text}
\usetikzlibrary{decorations.pathmorphing}
\usetikzlibrary{calc} 
\allowdisplaybreaks

\pgfplotsset{compat=1.18}

\usepackage{hyperref}
\usepackage{verbatim}

\definecolor{a0}{rgb}{0.61, 0.77, 0.89}
\definecolor{a1}{rgb}{1.0, 0.75, 0.0}
\definecolor{a2}{rgb}{1.0, 0.74, 0.53}

\begin{document}
\onehalfspace

\title{Complexity of Deciding the Equality of Matching Numbers}

\author{Guilherme C. M. Gomes$^1$ \and
Bruno P. Masquio$^2$\thanks{Partially supported by FAPERJ and CAPES} \and
Paulo E. D. Pinto$^2$ \and
Dieter Rautenbach$^4$ \and
Vinicius F. dos Santos$^1$\thanks{Partially supported by FAPEMIG and CNPq} \and
Jayme L. Szwarcfiter$^{2,3}$\thanks{Partially supported by FAPERJ and CNPq} \and
Florian Werner$^4$
}
\date{}

\maketitle
\vspace{-10mm}

\begin{center}
{\small 
  $^1$ Universidade Federal de Minas Gerais (UFMG) -- Belo Horizonte, MG -- Brazil 
  \texttt{$\{$gcm.gomes,viniciussantos$\}$@dcc.ufmg.br}\\[3mm]
  $^2$ Universidade do Estado do Rio de Janeiro (UERJ) -- Rio de Janeiro, RJ -- Brazil
  \texttt{$\{$brunomasquio,pauloedp$\}$@ime.uerj.br} \\[3mm]
  $^3$ Universidade Federal do Rio de Janeiro (UFRJ) -- Rio de Janeiro, RJ -- Brazil \\
  \texttt{jayme@nce.ufrj.br}\\[3mm]
  $^4$ Institute of Optimization and Operations Research -- Ulm University -- Ulm -- Germany
\texttt{$\{$dieter.rautenbach,florian.werner$\}$@uni-ulm.de}
}
\end{center}

\begin{abstract}
A matching is said to be disconnected if the saturated vertices induce a disconnected subgraph and induced if the saturated vertices induce a 1-regular graph. The disconnected and induced matching numbers are defined as the maximum cardinality of such matchings, respectively, and are known to be \NP-hard to compute.
In this paper, we study the relationship between these two parameters and the matching number. In particular, we discuss the complexity of two decision problems; first: deciding if the matching number and disconnected matching number are equal; second: deciding if the disconnected matching number and induced matching number are equal. 
We show that given a bipartite graph with diameter four, deciding if the matching number and disconnected matching number are equal is \NPc; the same holds for bipartite graphs with maximum degree three. 
We characterize diameter three graphs with equal matching number and disconnected matching number, which yields a polynomial time recognition algorithm.
Afterwards, we show that deciding if the induced and disconnected matching numbers are equal is \coNPc\ for bipartite graphs of diameter 3.
When the induced matching number is large enough compared to the maximum degree, we characterize graphs where these parameters are equal, which results in a polynomial time algorithm for bounded degree graphs.\\[3mm]
\end{abstract}

\newpage
\section{Introduction}
A \emph{matching} $M$ of a graph $G = (V,E)$
is a subset $M \subseteq E$ of the edges of $G$ such that no two edges in $M$ share a common endpoint. The set $V(M)$ contains all endpoints of edges of $M$, which are called \emph{saturated} vertices. In an abuse of notation, we write $G[M]$ for the induced subgraph $G[V(M)]$.
A matching $M$ is \emph{maximal} if there is no other matching $M'$ such that $M\subsetneq M'$ and \emph{maximum} if there is no other matching $M'$ of $G$ such that $|M'| > |M|$. The number of edges in a maximum matching of a graph $G$ is called the \emph{matching number}, denote by $\nu(G)$.
A matching saturating all vertices of $G$ is \emph{perfect}.

Matchings have been the subject of several studies in both structural and algorithmic graph theory.
It is a well-known fact that the size of a maximum matching can be computed efficiently \cite{edmonds_matching}.
Requiring further properties of a matching leads to restricted matchings:
A matching $M$ is said to be a $\mathscr{P}$-matching if $G[M]$ 
satisfies property $\mathscr{P}$. The $\mathscr{P}$ matching number is defined as the size of a maximum $\mathscr{P}$-matching.
The complexity of deciding whether or not a graph admits a $\mathscr{P}$-matching of given size has been investigated for many different properties $\mathscr{P}$ over the years, including
acyclic matchings \cite{10.1007/978-3-030-48966-3_31, GODDARD2005129}, 
degenerate matchings (generalization of acyclic matchings) \cite{BASTE201838}, induced matching \cite{CAMERON198997,LOZIN20027,MOSER2009715}, connected matchings \cite{cameron_connected_matching}, 
uniquely restricted matchings \cite{Golumbic2001} and disconnected matchings \cite{disconnected_matchings_cocoon, GODDARD2005129}. It appears natural to compare the $\mathscr{P}$ matching number of a given graph for different properties,
which results in the following decision problem:
Decide if, for two different properties, the corresponding $\mathscr{P}$ matching numbers of a given graph 
are equal. The following problems are known to be \NPH: Determining the equality of the induced and uniquely restricted matching numbers on bipartite graphs~\cite{on_the_hardness_of_deciding_the_equality_of_the_induced_and_the_uniquely_restricted_matching_number}, uniquely restricted and acyclic matching numbers~\cite{on_some_hard_and_some_tractable_cases_of_the_maximum_acyclic_matching_problem}, induced and acyclic matching numbers on graphs of maximum degree four~\cite{on_the_hardness_of_deciding_the_equality_of_the_induced_and_the_uniquely_restricted_matching_number},
and unrestricted matching and uniquely restricted matching numbers~\cite{on_some_hard_and_some_tractable_cases_of_the_maximum_acyclic_matching_problem}.
However, regarding comparing the induced matchings and unrestricted matchings, 
Kobler and Rotics
characterized the graphs with identical matching number and induced matching number,
yielding a polynomial time algorithm to recognize these graphs.
\cite{koblerrotics2003} 
\\

In this paper, we are particularly interested in (unrestricted) matchings, induced matchings and disconnected matchings.

A matching $M$ is said to be \textit{induced} if $G[M]$ is a 1-regular graph or empty. Deciding if a given graph admits an induced matching of a given size is a well-known \NPc\ problem~\cite{CAMERON198997}; the size of the largest induced matching in a graph $G$ is the \textit{induced matching number}, denoted by $\nu_s(G)$.
The hardness of computing this parameter extends to several graph classes, such as bipartite graphs with degree $4$~\cite{np_completeness_of_some_generalizations_of_the_maximum_matching_problem}, cubic
planar graphs~\cite{adding_an_identity_to_a_totally_unimodular_matrix}, $4k$-regular graphs for $k \geq 1$~\cite{induced_matchings_in_regular_graphs_and_trees}, subcubic bipartite graphs~\cite{on_maximum_induced_matchings_in_bipartite_graphs}, planar bipartite graphs where each vertex in one partition set has degree $2$ and each vertex in the other partition set has degree $3$~\cite{the_np_completeness_of_the_dominating_set_problem_in_cubic_planer_graphs}.
In~\cite{koblerrotics2003}, hardness is also shown for Hamiltonian graphs, claw-free graphs, chair-free graphs, line graphs, and $d$-regular graphs for $d \geq 5$. NP-completeness is shown for star-convex bipartite graphs and perfect elimination
in~\cite{induced_matching_in_some_subclasses_of_bipartite_graphs}.
Finally, the problem is shown to be \APXc\ in $d$-regular graphs for each $d \geq 3$ in~\cite{on_the_approximability_of_the_maximum_induced_matching_problem}.
On the other hand, efficient algorithms have been shown for several graph classes. Results include a linear-time algorithm for chordal graphs~\cite{Brandstdt2008} and polynomial time algorithms for weakly chordal graphs~\cite{finding_a_maximum_induced_matching_in_weakly_chordal_graphs}, permutation and trapezoid graphs~\cite{maximum_independent_set_and_maximum_induced_matching_problems_for_competitive_programming}, circular-convex bipartite graphs and triad-convex bipartite~\cite{induced_matching_in_some_subclasses_of_bipartite_graphs}, asteroidal-triple-free graphs ~\cite{induced_matchings_in_intersection_graphs,induced_matchings_in_asteroidal_triple_free_graphs}, interval-filament graphs~\cite{induced_matchings_in_intersection_graphs} and hexagonal graphs~\cite{maximum_induced_matching_of_hexagonal_graphs}.

A matching is said to be \textit{disconnected} if $G[M]$ is disconnected or empty.
The size of the largest disconnected matching in a graph $G$ is called the disconnected matching number, denoted by $\nu_d(G)$. The complexity of finding a disconnected matching of a given size was first asked in~\cite{GODDARD2005129}.
Motivated by this question, the authors of~\cite{disconnected_matchings_tcs} showed that deciding if a disconnected matching of a given size exists is \NPc.
In fact, they proved that finding a matching of size $k$ that induces a graph with at least $c$ connected components (a so called \textit{$c$-disconnected matching}) is \NPc\ for every fixed $c > 1$ even for bipartite graphs of diameter four.
For simplicity, we define the empty matching as $c$-disconnected for every $c\in \mathbb{N}$.
The \textit{$c$-disconnected matching number} is the size of the largest $c$-disconnected matching, denoted by $\nu_{d,c}(G)$. 
Note that $\nu_{d,i}(G) \geq \nu_{d, i+1}(G)$ for every
$i \geq 1$
and $\nu_{d, c}(G) > 0$ if and only if $c\le \nu_s(G)$.
For every graph $G$ and $\ell = \nu_s(G)$, it holds that:
\begin{align}
   \nu(G) = \nu_{d,1}(G) \geq \nu_{d,2}(G) \geq \dots \ge \nu_{d,\ell-1}(G)\geq \nu_{d,\ell}(G) = \nu_s(G)  \label{eq:nu_inequalities}
\end{align}
As mentioned before, deciding if $\nu_{s}(G) = \nu(G)$ is in \P \cite{koblerrotics2003}. 
However, the complexities of deciding the equality between the other (disconnected) matching numbers were unknown.
In this paper, we close that gap. Mainly, we consider two decision problems: First, deciding if the matching number and disconnected matching number are equal (decision problem $\nu=\nu_d$).
More general,
we consider for two fixed parameters $i$ and $j$ the decision problem $\nu_{d,i}=\nu_{d,j}$, i.e. deciding if the $i$-disconnected and the $j$-disconnected matching numbers are equal.
And second, deciding if
the disconnected matching number and induced matching number are equal (decision problem $\nu_d=\nu_s$).

First, we show \NPcness\text{} of the decision problem $\nu=\nu_d$ for two graph classes --
bipartite graphs with diameter four and 
bipartite graphs with maximum degree three. 
We extend this result with two corollaries:
For fixed $i\ge 2$, the decision problem $\nu = \nu_{d,i}$ is {\NPc} for bipartite graphs with diameter 4.
    For fixed $i$ and $j$ with $2\le i<j$, the decision problem $\nu_{d,i} = \nu_{d,j}$ is {\NPH} for bipartite graphs with diameter 3.
For graphs with diameter three, we characterize those with equal matching number and disconnected matching number and conclude that the decision problem $\nu=\nu_d$ is in \P\text{} for diameter three graphs.\\
Second, we show 
that the decision problem $\nu_d=\nu_s$ is \coNPc\ for bipartite graphs of diameter~3.
We characterizes all graphs $G$ with $\nu_s(G)=\nu_d(G)$ if $\nu_s(G)\ge 2 \Delta(G)$, resulting 
in a polynomial time algorithm for bounded degree graphs.
Table \ref{tab:p-matching-equalities} gives an overview of our main results.

\begin{table}[!htb]
\caption{Main results
\label{tab:p-matching-equalities}}
\begin{center}
\begin{tabular}{c|c|c|c}
\hline \hline
Problem & Complexity & Graph class & Proof \\ \hline

$\nu = \nu_d$ & {\NPc} & Bipartite and diameter 4 & Theorem~\ref{theoremnunud} \\ \hline

$\nu = \nu_{d}$ & {\NPc} & Bipartite and $\Delta=3$ & Theorem \ref{theoremnunuddelta3}\\ \hline

$\nu = \nu_d$ & {\P} & Diameter 3 & Theorem \ref{theoremdiam3} \\ \hline

$\nu_{s} = \nu_{d}$ & {\coNPc} & Bipartite and diameter at most 3& Theorem~\ref{nudnuscoNPcomplete} \\ \hline

$\nu_{s} = \nu_{d}$ & {\P} & Bounded degree & Theorem~\ref{theorempoly} \\ \hline

\hline
\end{tabular}
\end{center}
\end{table}

Finally, we show that, for every finite non-increase sequence of natural numbers $\beta_1, \dots, \beta_k$ with $\beta_k\ge k$ there exists 
a graph $G$ with $\nu_{d,i}(G) = \beta_i$ for every $i \in \{1,\ldots,k\}$.
This result implies that the differences between adjacent elements of the inequalities presented in Equation~\ref{eq:nu_inequalities} are arbitrary.\\

When there is no ambiguity of which graph we are referring to, we may write $\nu_s$, $\nu_d$, $\nu_{d,i}$ and $\nu$ instead of $\nu_s(G)$, $\nu_d(G)$, $\nu_{d,i}(G)$ and $\nu(G)$.

\section{Preliminaries}
We use standard nomenclatures and basic concepts of graph theory as in~\cite{murty,classes_survey}, complexity theory as in~\cite{garey_johnson}, and parameterized complexity as in~\cite{cygan_parameterized}.

For a set $C$, we say that $A,B \subseteq C$  is a \emph{partition} of $C$ if $A \cap B = \emptyset$ and $A \cup B = C$; we denote a partition of $C$ into $A$ and $B$ by $A \dot{\cup} B = C$. For an integer $k$, we define $[k] = \{1, \dots, k\}$. 

In this paper, we only use finite, simple, and undirected graphs. Let $G=(V,E)$ be a graph and $W\subseteq{V}$ a subset of its vertices. Sometimes, we also use $V(G)$ and $E(G)$ to denote the set of vertices and edges of $G$. Moreover, when there is no ambiguity, we use $n=|V(G)|$ and $m=|E(G)|$ for a graph $G$. An \emph{edge} $e$ is a pair of vertices $\{ u,v \}$ and $u,v$ are called \emph{endpoints} of $e$. We can equivalently write this edge as $uv$ or as $vu$. We say that $G[W]$ is the subgraph of $G$ \emph{induced} by $W \subseteq V$. That is, $G[W] = (W, E_W)$, such that $E_W$ contains an edge $e \in E$ if and only if $|e\cap{W}| = 2$. Also, the operations $G-uv$ and $G-v$ produce, respectively, the graphs $G'=(V,E\setminus\{uv\})$ and $G[V\setminus\{v\}]$. The \emph{degree} of vertex $v \in V$ is the number of edges of $E$ incident to $v$ and $\Delta(G)$ is the maximum vertex degree among all vertices of $G$. A graph $G$ is subcubic if $\Delta(G)\le 3$.

Two graphs $G$ and $H$ are \emph{isomorphic} if there is a bijection $f : V(G) \to V(H)$ such that $uv \in E$ if and only if $f(u)f(v) \in E(H)$. In $G$, a sequence of vertices $v_1\ldots{v_k}$ is a \emph{path} if $v_jv_{j+1} \in E(G)$ for every $1\leq j \leq k-1$. A \emph{cycle} is a path where $k\geq3$ and $v_k = v_1$. The \emph{length} of a cycle or a path is defined as the number of edges it contains. A graph is \emph{acyclic} if there is no induced subgraph isomorphic to a cycle. The \emph{distance} between two vertices $u,v$ is the length of the shortest path between $u$ and $v$. The \emph{diameter} of a connected graph $G$ is the longest distance between any pair of vertices $u,v \in V$. Furthermore, $G$ is \emph{connected} if there is a path between every pair of its vertices, and \emph{disconnected}, otherwise. A \emph{(connected) component}
of $G$ is a subgraph $G[W]$ for a maximal set $W \subseteq V$ subject to $G[W]$ is connected. 

The set $W$ is a \emph{separator} of $G$ if $G-W$ has more connected components than $G$. Besides, $W$ is \emph{minimal} if there is no other separator $S \subsetneq W$ in $G$. The \emph{open neighborhood} and \emph{closed neighborhood} of a vertex $u \in V$ are denoted by $N(u)$ and $N[u]$ respectively, where $N(u) = \{ w \mid wu \in E \}$ and $N[u] = N(u) \cup \{u\}$. Analogously, we define $N(W) = (\bigcup_{u \in W} N(u)) \setminus W$ and $N[W] = N(W) \cup W$.  

A graph $G$ is \emph{complete} if $E$ contains an edge for every pair of vertices of $V$. In this case, we can write this graph as $K_n$, $n = |V|$. A subset $S \subseteq V(G)$ is a \emph{clique} if $G[S]$ is complete. If all vertices of $G$ have degree $k$, then $G$ is \emph{$k$-regular}. We denote by $P_n$ and $C_n$ the \emph{path} and \emph{cycle graphs}, that are isomorphic to a path and a cycle with $n$ vertices, respectively. The set $W$ is an \emph{independent set} if $G[W]$ has no edges. A graph $G$ is \emph{bipartite} if its vertices can be partitioned into two independent sets $V_1$ and $V_2$. In this case, we also use the notation $G = (V_1\dot{\cup}V_2,E)$ and say $G$ is bipartite with partition $V_1\dot{\cup}V_2$. When $E$ contains all possible edges between elements of $V_1$ and $V_2$, we say that $G$ is a \emph{complete bipartite graph}. We denote by $K_{a,b}$ the complete bipartite graph with $a$ vertices in one set of the partition and $b$ vertices in the other.

\newpage

\section{\texorpdfstring{$\nu_d = \nu$?}{Disconnected matching number equals unrestricted}}
\begin{theorem}\label{theoremnunud}
Given a bipartite graph with diameter 4, deciding if $\nu=\nu_d$ is \NPc.
\end{theorem}

The decision problem $\nu=\nu_d$ is clearly in \NP.
In order to prove {\NPHness}, we describe a reduction 
from the {\NPc} problem { \sc Exact Cover By $3$-Sets} \cite{garey_johnson}. This problem consists in, given two sets $X = \{ x_1,\ldots,x_{3q} \}$ and $\mathcal{C}=\{C_1,\ldots,C_{|\mathcal{C}|}\}$, $|\mathcal{C}| \geq q$, such that $\mathcal{C}$ contains $3$-element subsets of $X$, decide if there exists a subset $\mathcal{C}' \subseteq \mathcal{C}$ such that every element of $X$ occurs in exactly one member of $\mathcal{C}'$. 
We call a subset $\mathcal{C}'\subseteq \mathcal{C}$ an exact cover of $X$, if $X$ is the disjoint union of the sets in $\mathcal{C}'$.\\
Given such an instance $(\mathcal{C},X)$ of { \sc Exact Cover By $3$-Sets}, we build the graph $G$ as follows:

\begin{enumerate}[label=(\Roman*), leftmargin=*, start=1]
    \item Generate a complete bipartite subgraph $H$ isomorphic to $K_{q,|\mathcal{C}|}$.
        Let $V_H$ be the partition set of size $|\mathcal{C}|$ of $H$ and 
        we label its vertices as $\{h_{j} : C_j \in \mathcal{C} \}$. Let $V'_H$ be the other partition set.
                
    \item For each $C_j \in \mathcal{C}$, generate a copy of $P_3$ whose endpoints are labeled $u_j^+$ and $u_j^-$. Connect the other vertex, $u_j$, to $h_j$.
    
    \item For each $x_i \in X$, generate the subgraph $Y_i$ isomorphic to $K_{f_i,f_i-1}$, where $f_i$ is the number of triples in $\mathcal{C}$ that contain the element $x_i$. Moreover, label the vertices of the bipartition of size $f_i$ as $\{ w_{i,j} : x_i \in C_j, C_j \in \mathcal{C}\}$. Add the edges $\{ u_iw_{i,j} : x_i \in C_j, C_j \in \mathcal{C} \}$. Let $W_i^+$ and $W_i^-$ be the set of vertices of bipartitions of size $f_i$ and $f_i-1$, respectively.
    
    \item Add two copies of $K_2$, whose vertices are labeled $\{ t^+,t^- \}$ and $\{ b^+, b^- \}$. Connect $b^-$ to all vertices in $\big\{ u_j^- : C_j \in \mathcal{C} \big\} \cup \bigcup_{x_i\in X} W_i^+$ and $t^-$ to all vertices in $V_H \cup \{ u_j^+ : C_j \in \mathcal{C} \}$.
\end{enumerate}
See Figure~\ref{fig:nunud-ex-4-mv} for an example of the reduction. Note that $G$ is indeed bipartite with partition sets 
$V_1 = \bigcup_{x_i\in X } W_i^- \cup \{ u_j : C_j \in \mathcal{C}\} \cup V_H' \cup \{ b^-,t^- \}$ and 
$V_2 = \bigcup_{C_j\in \mathcal{C}} N(u_j)\cup \{ t^+,b^+ \}$. 

\begin{figure}[!htb]
	\centering
	\begin{tikzpicture} [scale=0.275] 	\tikzstyle{point}=[draw,circle,inner sep=0.cm, minimum size=1mm, fill=none]
	\tikzstyle{point2}=[draw,circle,inner sep=0.cm, minimum size=1mm, fill=black]

\node[point2] (t) at (40,2) [label=right:$t^-$] {};
\node[point2] (tt) at (40,-2) [label=right:$t^+$] {};
\draw[very thick] (t) -- (tt);

\node[point2] (b) at (40,-29) [label=right:$b^-$] {};
\node[point2] (bp) at (40,-25) [label=right:$b^+$] {};
\draw[very thick] (b) -- (bp);

\foreach \j in {0,1,2,3}{
    \pgfmathparse{int(\j+1)}
    \xdef\jj{\pgfmathresult}
  	\node[point] (h\j) at (10*\j,-2) [label=left:$h_\jj$] {};
}

\foreach \j in {0,2}{
		\node[point] (p\j) at (5*\j+10,2) [label=above:] {};
}

\foreach \i in {0,1,2,3}{
    \foreach \j in {0,2}{
        \draw (h\i) -- (p\j);
    }
}

\foreach \i in {0,1,2,3}{
    \draw (h\i) -- (t);
}

\foreach \l in {0,1,2,3}{

\begin{scope}[shift={(10*\l,-6)}]

\pgfmathparse{int(\l+1)}
\xdef\ll{\pgfmathresult}
\node[point] (u\l) at  (0,0) [label=right:$u_\ll$] {};
\node[point] (uplus\l) at  (-2,1) [label=left:$u_\ll^+$] {};
\node[point] (uminus\l) at  (-2,-1) [label=left:$u_\ll^-$] {};
\draw (h\l) -- (u\l) -- (uplus\l) -- (t); 
\draw (u\l) -- (uminus\l) -- (b); 

\end{scope}
}

\def\myarray{{3,3,1,2,2,1}}
\edef\m{0}

\edef\sh{0}

\foreach \j in {0,1,2,3,4,5}{
	\pgfmathparse{int(\myarray[\m]-1)}
	\let\result\pgfmathresult
	\pgfmathparse{\m+1}
	\xdef\m{\pgfmathresult}
	
	\begin{scope}[shift={(-4+\sh,-25)}]
 	\pgfmathparse{\sh+3*(\result+1)+1}
	\xdef\sh{\pgfmathresult}
	\foreach \i in {0,...,\result}{
       \node[point] (w\j\i) at (3*\i,0) [label=above:] {};
       \draw (w\j\i) -- (b);	}
    \if\result0\else
    \foreach \i in {1,...,\result}{
            \node[point] (o\j\i) at (3*\i-1.5,-4) [label=above:] {};
    }
    \foreach \i in {0,...,\result}{
        \foreach \k in {1,...,\result}{
                \draw (w\j\i) -- (o\j\k);
            }
        }
    \fi
    \end{scope}
}

 \def\myarr{{0,1,2,0,1,3,0,1,4,3,4,5}}
\def\myarrt{{0,0,0,1,1,0,2,2,0,1,1,0}}
\edef\s{0}

\foreach \i in {0,1,2,3}{
	\foreach \j in {1,2,3}{
		\pgfmathparse{\myarr[\s]}
		\let\x\pgfmathresult
		\pgfmathparse{\myarrt[\s]}
		\let\y\pgfmathresult
		\draw (u\i) -- (w\x\y);
        \pgfmathparse{int(\x+1)}
        \xdef\xr{\pgfmathresult}
        \pgfmathparse{int(\i+1)}
        \xdef\yy{\pgfmathresult}
        \draw (w\x\y |- 42,-24) node {$w_{\xr,\yy}$};
		\pgfmathparse{\s+1}
		\xdef\s{\pgfmathresult}
	}
}

\draw[very thick] (h0) -- (p0);
\node[point2] at  (h0) [label=right:] {};
\node[point2] at  (p0) [label=right:] {};
\draw[very thick] (h3) -- (p2);
\node[point2] at  (h3) [label=right:] {};
\node[point2] at  (p2) [label=right:] {};
\foreach \i in {0,3}{
    \draw[very thick] (u\i) -- (uplus\i);
    \node[point2] at  (u\i) [label=right:] {};
    \node[point2] at  (uplus\i) [label=right:] {};
}
\foreach \i in {1,2}{
    \draw[very thick] (u\i) -- (uminus\i);
    \node[point2] at  (u\i) [label=right:] {};
    \node[point2] at  (uminus\i) [label=right:] {};
}

\foreach \i in {0,1}{
    \foreach \j in {1,2}{
        \draw[very thick] (w\i\j) -- (o\i\j);
        \node[point2] at  (w\i\j) [label=right:] {};
        \node[point2] at  (o\i\j) [label=right:] {};
    }
}
\foreach \i in {3,4}{
    \foreach \j in {1}{
        \draw[very thick] (w\i0) -- (o\i\j);
        \node[point2] at  (w\i0) [label=right:] {};
        \node[point2] at  (o\i\j) [label=right:] {};
    }
}

	\end{tikzpicture}
    \caption{
Reduction graph for the input $X = \{x_1,\ldots,x_6\}$ and $\mathcal{C} = \{C_1,C_2,C_3,C_4\}= \{ \{x_1,x_2,x_3\},\\ \{ x_1,x_2,x_4 \}, \{ x_1,x_2,x_5 \}, \{ x_4,x_5,x_6 \} \}$, 
including a maximum matching corresponding to the exact cover $\mathcal{C}'=\{ \{x_1,x_2,x_3\}, \{x_4,x_5,x_6\} \}$.
        }
    \label{fig:nunud-ex-4-mv}
\end{figure}
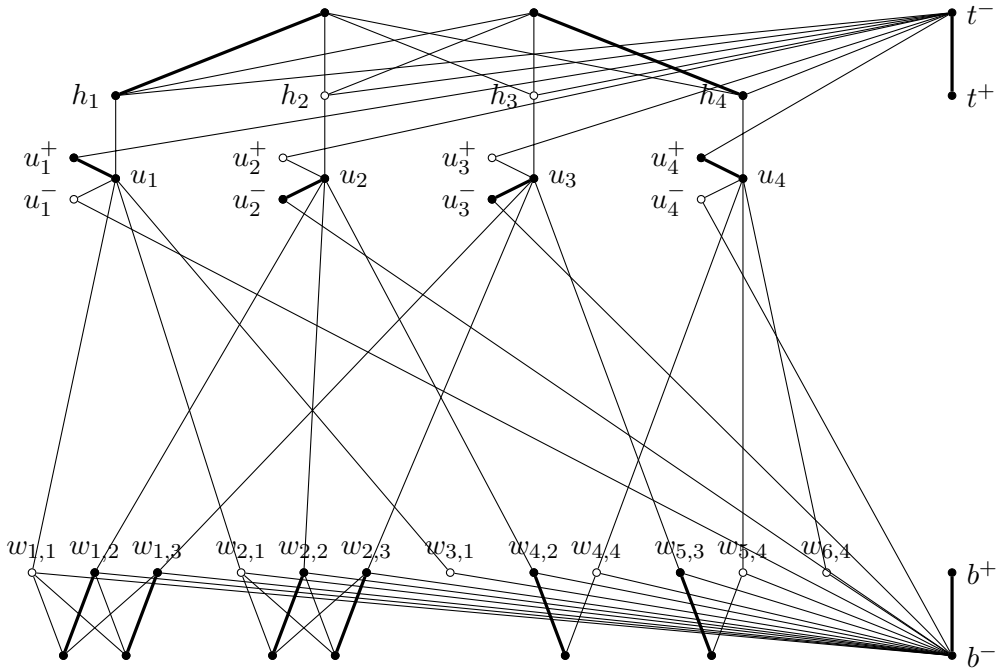

\begin{lemma}\label{lem:nud-nu-size}
A maximum matching of $G$ has size $|V_1|=4|\mathcal{C}|-2q+2$ and saturates all vertices of $V_1$.
\end{lemma}
\begin{proof}
Since $G$ is bipartite and $|V_1|=4|\mathcal{C}|-2q+2$, it is enough to define a matching that saturates $V_1$.
Let $M$ be the union of four sets, $S_1 = \{ b^+b^-, t^+t^-\}$, $S_2 = \{ u_ju_j^+ : C_j \in \mathcal{C}\}$, $S_3$ consisting of
$q$ disjoint edges of $H$, and $S_4$ 
containing
$f_i-1$ disjoint edges of $Y_i$ for each $x_i \in X$. The matching $M$ saturates all vertices in $V_1$, completing the proof of Lemma \ref{lem:nud-nu-size}.
\end{proof}

Note that for every maximum matching $M$ the graph $G[M]$ 
has
at most two components, since 
every vertex of $G$ is adjacent to $t^-$ or $b^-$ and $t^-, b^-\in V_1$.

\begin{lemma}\label{lem:nudnu-main-lemma}
The instance $(X,\mathcal{C})$ of \pname{Exact Cover By 3-sets} is a \YES-instance if and only if 
$\nu(G)=\nu_d(G)$. 
\end{lemma}
\begin{proof}
($\Rightarrow$): 
Let $(X,\mathcal{C})$ be a \YES-instance of \pname{Exact Cover By 3-sets} and let $\mathcal{C}' \subseteq \mathcal{C}$ be an exact cover of $X$. We construct the matching $M$ as follows.
\begin{enumerate}
\item Add the edges $b^-b^+$ and $t^-t^+$.
\item Choose $q$ many edges from $H$ that saturate all vertices in $\{ h_j: C_j\in \mathcal{C}' \}$.
\item For every $C_j\in \mathcal{C}$, we choose the edge $u_ju_j^+$ if $C_j\in \mathcal{C}'$ and $u_ju_j^-$ otherwise.
\item    For every 
$x_i\in X$,
choose the edges of a perfect matching of the complete bipartite graph $Y_i - w_{i,j} $, where $j$ is the unique index such that $x_i\in C_j$ and $C_j\in \mathcal{C}'$.
\end{enumerate}
Since $M$ saturates all vertices in $V_1$, it is a maximum matching. To show that $M$ is disconnected, consider the set 
$U\coloneqq \{t^-,t^+\} \cup \bigcup_{C_j \in \mathcal{C}'} \{ h_j,u_j,u_j^+ \} \cup V'_H$, containing only saturated vertices.
Since $N_G[U]\setminus U=\left \{  w_{i,j}: x_i\in C_j\in \mathcal{C}' \right \} \cup \{ u_j^-: C_j\in \mathcal{C}' \} \cup \{  h_j: C_j\in \mathcal{C}\setminus \mathcal{C}' \}  $, the vertices in $U$ have no neighbors outside of $U$ that are saturated.
Hence $M$ is indeed disconnected, where one component contains the vertices of $U$ and the other component the remaining saturated vertices.\\
\newline
($\Leftarrow$): Let $\nu(G)=\nu_d(G)$. Lemma~\ref{lem:nud-nu-size} implies the existence of a disconnected matching $M$ that saturates all vertices of $V_1$. 
The vertices $t^-$ and $b^-$ are contained in $V_1$ and hence saturated by $M$. 
Every edge of $G$ is adjacent to either $t^-$ or $b^-$, implying $G[M]$ has exactly two connected components, each containing one of those two vertices.

Since $V_H'\subseteq V_1$, $M$ saturates all vertices of $V_H'$, implying that at least $q$ many vertices of $V_H$ are saturated. 
Let $J=\{j: C_j\in \mathcal{C}, h_j \in V(M) \} $ 
be the indices of the saturated vertices in $\{h_1,\ldots, h_{|\mathcal{C}|} \}=V_H$, where $|J|\ge q$. 
Since $u_j\in V_1$ for every $C_j\in \mathcal{C}$, $u_j$ is saturated by $M$. Hence for all $j\in J$, $u_j$ is in the same component of $G[M]$ as 
$t^-$.

For every $j\in J$, none of the vertices in $N(u_j)\cap \left( \bigcup_{x_i\in X} W_i^+ \right)$ can be saturated by $M$, since otherwise $G[M]$ would contain a path from $t^-$ to $b^-$.
Hence there are at least $3|J|$ many non-saturated vertices in $\bigcup_{x_i \in X} W_i^+$.

On the other hand, for every $x_i\in X$, $W_i^-\subseteq V_1$, and since $|W_i^-|=|W_i^+|-1$, at most one vertex in $W_i^+$ is not saturated by $M$. Together, at most $3q$ many vertices in $\bigcup_{x_i \in X} W_i^+$ are not saturated.

The previous two arguments imply that exactly $3q$ many vertices in $\bigcup_{x_i \in X} W_i^+$ are not saturated -- one per $W_i^+$ -- and those vertices are $ \bigcup_{j\in J} \left[N(u_j)\cap \left( \bigcup_{x_i\in X} W_i^+ \right) \right]$.
Hence $\bigcup_{j\in J} C_j=X$ and $|J|=q$, implying $\mathcal{C}'=\{ C_j:j\in J \}$ is an exact cover of $X$.
\end{proof}

  To receive a graph with diameter 4, add a vertex $u$ and connect it to all vertices in $V_1$. Lets call this graph $G'$.
  Since every vertex is adjacent to $u$ or has a neighbor adjacent to $u$, $G'$ has diameter 4.
   Note that $\nu(G')=\nu(G)=|V_1|$ and since the vertex $u$ can't be saturated by a disconnected matching, $\nu_d(G')=\nu_d(G)$. Hence $\nu(G')=\nu_d(G')$ if and only if $\nu(G)=\nu_d(G)$. This concludes the proof of Theorem \ref{theoremnunud}.
 \begin{lemma}\label{cornunudi}
For fixed $i\ge 2$, deciding if $\nu = \nu_{d,i}$ is {\NPc} for bipartite graphs with diameter 4.
\label{corollarynunudi}
\end{lemma}

\newcommand{\Gnew}{G_1}

\begin{proof}
    The decision problem is clearly in \NP. To show \NPHness, 
we adjust the graph $G$ by adding $i-2$ copies of $K_2$ consisting of the vertices $u_1,\ldots, u_{i-2}$, $v_1,\ldots, v_{i-2}$ and the edges $u_1v_1,\ldots, u_{i-2}v_{i-2}$. Further more, we add a vertex $u$ adjacent to $V_1\cup \{v_1,\ldots, v_{i-2}\}$, resulting in a new graph $\Gnew$. The graph $\Gnew$ has bipartition $V_1' \dot{\cup} V_2'$, where $V_1'=V_1 \cup \{v_1,\ldots, v_{i-2}\}$ and 
$V_2'=V_2 \cup \{u,u_1,\ldots, u_{i-2} \}$. 
Since every vertex in $V(\Gnew)\setminus \{u\}$ is either adjacent to $u$ or has a neighbor adjacent to $u$,
the graph $\Gnew$ has diameter 4.
Obviously $\nu(\Gnew)=|V_1''|=\nu(G)+(i-2)$.
Observe that any matching in $\Gnew$ that saturates the vertex $u$ is connected.
Therefore, for $i\ge 2$, any maximum $i$-disconnected matching of $G_1$ consists of a maximum disconnected matching of $G$ and the edges $u_1v_1,\ldots, u_{i-2}v_{i-2}$, implying $\nu_{d,i}(\Gnew)= \nu_d(\Gnew)+i-2$. Since $\nu(\Gnew)=\nu(G)+(i-2)$ and $\nu_{d,i}(\Gnew)= \nu_d(\Gnew)+i-2$, it holds that $\nu(\Gnew)=\nu_{d,i}(\Gnew)$ if and only if $\nu(G)=\nu_d(G)$. Applying Lemma~\ref{lem:nudnu-main-lemma} concludes the proof of Lemma~\ref{cornunudi}.\\
\end{proof}

\begin{lemma}\label{cornuinuj}
For fixed $i$ and $j$ with $2\le i<j$,
deciding if $\nu_{d,i} = \nu_{d,j}$ is {\NPH} for bipartite graphs with diameter 3.
\end{lemma}

\newcommand{\Gneww}{G_2}

\begin{proof}
We adjust the graph $G$ by adding $j-2$ copies of $K_2$ consisting of the vertices $u_1,\ldots, u_{j-2}$, $ v_1,\ldots, v_{j-2}$ and the edges $u_1v_1,\ldots, u_{j-2}v_{j-2}$
and by adding another copy of $K_2$ with vertices $u,v$. We connect $u$ to every vertex in $V_1'\coloneqq V_1\cup \{v_1,\ldots, v_{j-2}\}  $ and $v$ to every vertex in $V_2'\coloneqq V_2 \cup \{u_1,\cdots, u_{j-2} \} $, resulting in a new bipartite graph $\Gneww$ with the partition classes $V_1'$ and $V_2'$ and diameter 3.
Observe that any matching that saturates $u$ or $v$ is connected. Hence $\nu_{d,i}(\Gneww)
=\nu_{d,i}(\Gneww-u-v)$ and $\nu_{d,j}(\Gneww)=\nu_{d,j}(\Gneww-u-v)$. 
The graph $\Gneww-u-v$ consists of $j-1$ components, namely $j-2$ many $K_2$ components and $G$.
It follows immediately, that $\Gneww-u-v$ has a $j$-disconnected matching of size $\nu(\Gneww)$ if and only if $\nu(G)=\nu_d(G)$ and hence $\nu_{d,i}(\Gneww)= \nu_{d,j}(\Gneww)$ if and only if $\nu(G)=\nu_d(G)$.
\end{proof}

\section{\texorpdfstring{$\nu_d = \nu$ for $\Delta\le 3$?}{Disconnected equals unrestricted for maximum degree 3}}

\newcommand{\Gdthree}{F}

\newcommand{\btmn}[2]{ H(#1,#2) }

\newcommand{\mg}[1]{ U_{#1} }

\newcommand{\mrb}{ \Pi_b }

\newcommand{\mrt}{ \Pi_t }

\newcommand{\tminus}{ t^- }
\newcommand{\ttt}{ t}
\newcommand{\bbb}{ b}
\newcommand{\tplus}{ t^+ }
\newcommand{\bminus}{ b^- }
\newcommand{\bplus}{ b^+ }

\newcommand{\Qt}{ Q }
\newcommand{\Qb}{ Q' }

\newcommand{\Cabs}{|\mathcal{C}|}

\newcommand{\ccc}[2]{\Pi(#2,#1)}
\newcommand{\ccs}[1]{\Pi(#1)}

\newcommand{\ua}[2]{ u^{#1}_{#2,1} }
\newcommand{\ub}[2]{ u^{#1}_{#2,2} }

\newcommand{\down}[2]{ u^{#1}_{#2,3} }
\newcommand{\downb}[2]{ u^{#1}_{#2,4}}

\newcommand{\plus}[1]{ u^+_{#1} }
\newcommand{\minus}[1]{ u^-_{#1} }

\newcommand{\Ynew}[1]{Y'_{#1} }

\newcommand{\Aa}{A}
\newcommand{\Bb}{B}
\newcommand{\Cc}{C}   \newcommand{\Dd}{D}

\newcommand{\BB}{B}
\newcommand{\CC}{C}   \newcommand{\DD}{D}

\newcommand{\ax}{a}
\newcommand{\bx}{b}
\newcommand{\cx}{c}
\newcommand{\cstar}[1]{c^*_{#1}}
\newcommand{\dx}{d}

In the previous chapter we described a reduction from {\sc Exact Cover By 3-Sets} to the decision problem $\nu_d(G)=\nu(G)$, where an instance $(\mathcal{C}, X)$ was reduced to a bipartite graph $G$.
Now, we refine our previous reduction such that the graph generated is not only bipartite but also subcubic.
This implies the following:

\begin{theorem}\label{theoremnunuddelta3}
Given a subcubic bipartite graph,
deciding if $\nu = \nu_d$ is \NPc.
\end{theorem}

Let $(\mathcal{C},X)$ be an instance of  {\sc Exact Cover By 3-Sets}. Recall that $|X|=3q$. Note that we can assume $|\mathcal{C}|\ge q\ge  2$, since otherwise the instance is efficiently solvable. We construct a graph $\Gdthree$ of maximum degree three such that $\nu_d(\Gdthree)=\nu(\Gdthree)$ if and only if $(\mathcal{C}, X)$ is a YES-instance.
For better understandability, we first show graphically the transformation from $G$ to $\Gdthree$, where $G$ is the graph of the previous chapter according to (I) -- (IV) for the instance $(\mathcal{C},X)$. 
Afterwards, we formally define $\Gdthree$ and validate the reduction.\\

\newcommand{\tstar}{ t^{*}}
\newcommand{\bstar}{ b^{*}}
\newcommand{\tstars}{ t^{**}}
\newcommand{\bstars}{ b^{**}}

Initially, we choose $\Gdthree $ as a copy of $G$ and modify it step-by-step.
The degree of the vertices $\tminus$ and $\bminus$ is of course too large. We are going to replace the edges $\tminus\tplus$ and $\bminus\bplus$ with subgraphs. For simplicity, this will be the last step. To avoid confusion, label $\tminus,\tplus, \bminus, \bplus$ in $\Gdthree$ as $\tstar,\tstars, \bstar,\bstars$, respectively.

\subsection{\texorpdfstring{Replace $H$ isomorphic to $K_{q,\Cabs}$ with the subgraph $\btmn{q}{\Cabs}$}{Modification part 1}}\label{secbm}

In $\Gdthree$, we replace the complete bipartite subgraph $H$ isomorphic to $K_{q,|\mathcal{C}|}$ with a new subgraph $\btmn{q}{|\mathcal{C}|} $ (see example in Figure \ref{H35}).\\

The graph $\btmn{q}{|\mathcal{C}|}$ is constructed as follows:
\begin{enumerate}[label=(\roman*), leftmargin=*, start=1]
	\item For each $k\in [q]$ generate a path $Q_k$ isomorphic to $P_{2|\mathcal{C}|-1}$. Declare one of its end-vertices as the     first vertex implying a natural order of the vertices.
	
	\item For each $j\in \left[|\mathcal{C}| \right]$ generate a path $R_j$ isomorphic to $P_{2q-1}$ 
  with ``first'' vertex $\cstar{j}$, implying a natural order of the vertices of $R_j$.
	
	\item For each $j\in [|\mathcal{C}|]$ and $k\in [q]$ add an edge between the 
  $(2j-1)$th vertex of $Q_k$ and the 
  $(2k-1)$th vertex of $R_j$.
\end{enumerate}
Let the set $\Bb$ contain the
first, third, fifth, \ldots vertex of each path $Q_1,\ldots, Q_q$
and $\Aa$ all the other vertices of $Q_1,\ldots, Q_q$. Let $\Cc$ contain the first, third, fifth, \ldots vertex of each path $R_1,\ldots, R_{\Cabs}$ and $\Dd$ all the other vertices of $R_1,\ldots, R_{\Cabs}$.\\
Figure \ref{HqC35} shows a $\btmn{3}{5}$, where $\Aa,\Bb,\Cc,\Dd$ contain the
vertices of the first, second, third and fourth row, respectively.

To receive $\btmn{q}{\Cabs}$ from $K_{q,|\mathcal{C}|}$, graphically and simplified speaking, we replace the vertices of degree $q$ in $K_{q,|\mathcal{C}|}$ with paths of length $2q-1$, we replace the vertices of degree $|\mathcal{C}|$ in  $K_{q,|\mathcal{C}|}$ with paths of length $2|\mathcal{C}|-1$ and we distribute the original edges between the first, third, fifth, $\ldots$ vertices of those paths. 

\begin{figure}[htb!]
	\centering
	\begin{tikzpicture}[scale=0.37]
	\tikzstyle{point}=[draw,circle,inner sep=0.cm, minimum size=1mm, fill=black]
	
	\foreach \j in {0,1,2,3,4}{
		
		\begin{scope}[shift={(7*\j,-2)}]
		
		\node[point] (w) at (0,0) [label=above:] {};
		\foreach \i in {0}{
			\node[point] (h\j\i) at (0,0) [label=above:] {};
		}
		
		\foreach \i in {1,...,2}{
			\node[point] (ww) at (2*\i-1,-1) [label=above:] {};
			\draw (w) -- (ww);
			\node[point] (w) at (2*\i,0) [label=above:] {};
			\node[point] (h\j\i) at (2*\i,0) [label=above:] {};
			\draw (w) -- (ww);
		}
		
		\end{scope}
	}
	
	\foreach \j in {0,1,2}{
		
		\begin{scope}[shift={(12*\j,0)}]
		
		\node[point] (w) at (0,0) [label=above:] {};
		\foreach \i in {0}{
			\node[point] (a\j\i) at (0,0) [label=above:] {};
		}
		\foreach \i in {1,...,4}{
			\node[point] (ww) at (2*\i-1,1) [label=above:] {};
			\draw (w) -- (ww);
			\node[point] (w) at (2*\i,0) [label=above:] {};
			\node[point] (a\j\i) at (2*\i,0) [label=above:] {};
			\draw (w) -- (ww);
		}
		
		\end{scope}
	}
	
	\foreach \i in {0,...,4}{
		\foreach \j in {0,...,2}{
			\draw (h\i\j) -- (a\j\i);
		}
	}

	\end{tikzpicture}
	\caption{$\btmn{3}{5}$}
	\label{HqC35}
\end{figure}

After replacing $H$ with $\btmn{q}{\Cabs}$, we adjust the adjacent edges (see Figure \ref{H35}). In $G$ each vertex $h_j$ of the partition of size $\Cabs$ of $H$ is adjacent to $\tminus$ and $u_j$. 
In $\Gdthree$ we connect every vertex in $\Aa$ of $\btmn{q}{\Cabs}$ to $\tstar$ and for each $j\in [\Cabs]$ we add an edge between $\cstar{j}$ and $u_j$. 
	Note that we modify $\Gdthree$ step-by-step, the vertex $\tstar$ is not in the final graph.

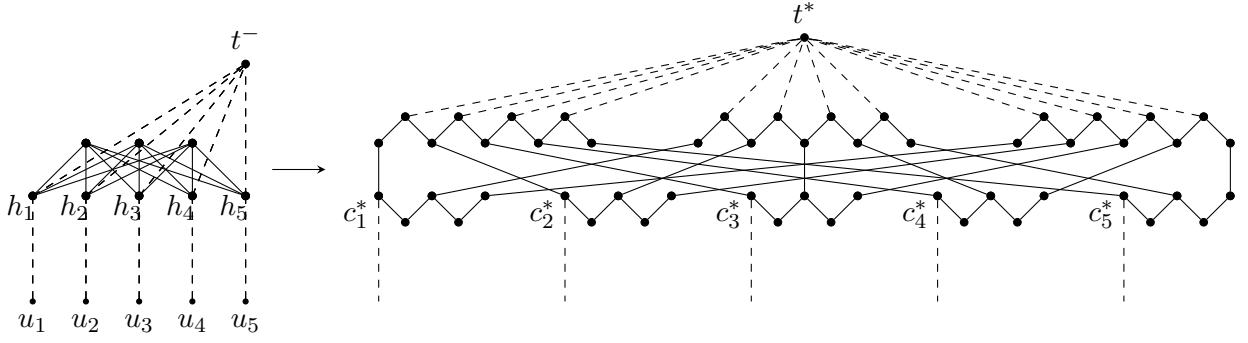
\begin{figure}[htb!]
	\centering
	\begin{tikzpicture}[scale=0.35]
	
	\tikzstyle{point}=[draw,circle,inner sep=0.cm, minimum size=1.0mm, fill=black]
	\tikzstyle{point2}=[draw,circle,inner sep=0.cm, minimum size=0.666mm, fill=black]
	
	\foreach \i in {-5,-4,-3}{
		\foreach \j in {-6,-5,-4,-3,-2}{
		\node[point] (a) at (2*\i,-2) [label=above:] {};
		\node[point] (a) at (2*\j,-4) [label=below left:] {};
		\draw (2*\i,-2) -- (2*\j,-4);
		\draw[dashed] (2*\j,-4) -- (2*\j,-8);
		\draw[dashed] (2*\j,-4) -- (-4,1);
	}
}
	\foreach \j in {1,2,3,4,5}{
	\node (a) at (2*\j-2*6.55,-3.2) [label=below left:$h_{\j}$] {};
	\node[point2] (a) at (2*\j-2*7,-8) [label=below:$u_{\j}$] {};
}

	\node[point] (a) at (-4,1) [label=above:$\tminus$] {};
	
	\draw [-stealth](-3,-3) -- (-1,-3);
	
	\begin{scope}[shift={(1,-2)}]
	\foreach \j in {0,1,2,3,4}{
		
		\begin{scope}[shift={(7*\j,-2)}]
		
		\node[point] (w) at (0,0) [label=above:] {};
		\foreach \i in {0}{
			\node[point] (h\j\i) at (0,0) [label=above:] {};
		}
		
		\foreach \i in {1,...,2}{
			\node[point] (ww) at (2*\i-1,-1) [label=above:] {};
			\draw (w) -- (ww);
			\node[point] (w) at (2*\i,0) [label=above:] {};
			\node[point] (h\j\i) at (2*\i,0) [label=above:] {};
			\draw (w) -- (ww);
		}
		
		\end{scope}
	}
	
	\foreach \j in {0,1,2}{
		
		\begin{scope}[shift={(12*\j,0)}]
		
		\node[point] (w) at (0,0) [label=above:] {};
		\foreach \i in {0}{
			\node[point] (a\j\i) at (0,0) [label=above:] {};
		}
		\foreach \i in {1,...,4}{
			\node[point] (ww) at (2*\i-1,1) [label=above:] {};
			\node[point] (b\j\i) at (2*\i-1,1) [label=above:] {};
			\draw (w) -- (ww);
			\node[point] (w) at (2*\i,0) [label=above:] {};
			\node[point] (a\j\i) at (2*\i,0) [label=above:] {};
			\draw (w) -- (ww);
		}
		
		\end{scope}
	}
	
	\foreach \i in {0,...,4}{
		\foreach \j in {0,...,2}{
			\draw (h\i\j) -- (a\j\i);
		}
	}

	\foreach \i in {0,1,2,3,4}{
		\foreach \j in {0}{
			\pgfmathparse{int(\i+1)}
			\def\r{\pgfmathresult}
			\node (a) at (h\i\j) [label={[shift={(-0.3,-0.7)}]:$\cstar{\r}$}] {};
			\draw[dashed] (h\i\j) -- + (0,-4);
		}
	}
	
	\node[point] (t) at (16,4) [label=above:$\tstar$] {};
	
	\foreach \i in {1,2,3,4}{
		\foreach \j in {0,1,2}{
			\draw[dashed] (b\j\i) -- (t);
		}
	}
	
	\end{scope}
	
	\end{tikzpicture}
	\caption{Example for $q=3$ and $\Cabs=5$: $H\cong K_{3,5}$ $\rightarrow$ $\btmn{3}{5}$}
	\label{H35}
\end{figure}

\newcommand{\va}[1]{v_{#1,1}}
\newcommand{\vb}[1]{v_{#1,2}}
\newcommand{\vc}[1]{v_{#1,3}}
\newcommand{\vd}[1]{v_{#1,4}}
In the following subsections, we use the following graph:
\begin{definition}\label{definitionccc}
Given $k\in \mathbb{N}$ and a label $v$, the graph $\ccc{v}{k}$ has the vertices $\va{\ell},\vb{\ell},\vc{\ell},\vd{\ell}$ for every $\ell \in [k]$. Afterwards, we add $2k-1$ many edges such that $\va{1}\vb{1}\va{2}\vb{2}\cdots \va{k}\vb{k}$ is a path of length $2k$ and $2k$ many edges such that $\vb{\ell}\vc{\ell}\vd{\ell}$ is a path of length three for every $\ell\in [k]$ (see Figure \ref{masternew}).
\end{definition}

	\begin{figure}[htb!]
		\centering
		\begin{tikzpicture}[scale=0.9]
		
		\newcommand{\xx}[2]{ \vb{#2} }
		\newcommand{\xy}[2]{ \va{#2} }
		\newcommand{\xz}[2]{ \vc{#2} }
		\newcommand{\xw}[2]{ \vd{#2} }
		
		\tikzstyle{point}=[draw,circle,inner sep=0.cm, minimum size=1.5mm, fill=black]

		\draw (1,1) -- (8,1);
		\foreach \i in {1,2,3,4  }{
			\node[point] (x) at (2*\i,1) [label=above:$\xx{j}{\i}$] {};
			\node[point] (y) at (2*\i-1,1) [label=above:$\xy{j}{\i}$] {};
			\node[point] (z) at (2*\i,0) [label=below:$\xz{j}{\i}$] {};
			\node[point] (w) at (2*\i-1,0) [label=below:$\xw{j}{\i}$] {};
			\draw (x) -- (z) -- (w);
		}
		
		\tikzstyle{point2}=[draw,circle,inner sep=0.cm, minimum size=0.5mm, fill=black]
		
		\draw (8,1) -- (8.5,1);
		
		\foreach \i in {9,9.5,10 }{
			\node[point2] at (\i,1) [label=right:] {};
		}
		
		\draw (10.5,1) -- (12,1);

		\foreach \i in {6  }{
			\node[point] (x) at (2*\i,1) [label=above:$\xx{j}{k}$] {};
			\node[point] (y) at (2*\i-1,1) [label=above:$\xy{j}{k}$] {};
			\node[point] (z) at (2*\i,0) [label=below:$\xz{j}{k}$] {};
			\node[point] (w) at (2*\i-1,0) [label=below:$\xw{j}{k}$] {};
			\draw (x) -- (z) -- (w);
		}	
		
		\end{tikzpicture}
		\caption{$\ccc{k}{v}$}
		\label{masternew}
	\end{figure}

\subsection{\texorpdfstring{Replace the $P_3$ $\plus{j}u_j\minus{j}$ with the subgraph $\mg{j}$}{Modification part 2}}\label{secmg}

In $G$, for each $C_j\in \mathcal{C}$, the vertex $u_j$ has the neighbors $h_j$, $u_j^+$, $u_j^-$ and three neighbors in $ \bigcup_{x_i\in X} W_i^+ $, namely the vertices $w_{i,j}$ for all $x_i\in C_j$. In $\Gdthree$, we replace the $P_3$ $\plus{j}u_j\minus{j}$ with the subgraph $\mg{j}$, which contains a $\ccc{u^j}{4}$, a $P_3$ $\plus{j}\ua{j}{5}\minus{j}$ and the edge $\ub{j}{4}\ua{j}{5}$ (see Figure \ref{transformationuj}).

Now we add the edges $\plus{j}\tstar$ and $\minus{j}\bstar$, like in $G$. We connect $\down{j}{4}$ to $\cstar{j}$ and we add three edges such that each vertex in $\big\{ \down{j}{1}, \down{j}{2}, \down{j}{3}\big\}$ is adjacent to exactly one vertex in $\{w_{i,j}: x_i\in C_j\}$ and vice versa. 

\begin{figure}[htb!]
	\centering
	\begin{tikzpicture}[scale=0.68]
	
	\newcommand{\xx}[2]{}
	\newcommand{\xy}[2]{}
	\newcommand{\xz}[2]{\down{#1}{#2} }
	\newcommand{\xw}[2]{}
	
	\tikzstyle{point}=[draw,circle,inner sep=0.cm, minimum size=1.5mm, fill=black]
	\tikzstyle{point2}=[draw,circle,inner sep=0.cm, minimum size=1mm, fill=black]
	
	\foreach \i in {-5  }{
		\node[point] (uj) at (\i,1) [label=above left:$u_j$] {};
		\node[point] (ujplus) at  (\i+1,2) [label=above:$u_j^+$] {};
		\node[point] (ujminus) at  (\i+1,0) [label=above:$u_j^-$] {};
		\draw (uj) -- (ujplus); 
		\draw (uj) -- (ujminus); 
		\draw[dashed] (ujplus) -- (\i+3,4);
		\draw[dashed] (ujminus) -- (\i+3,-2);
		\draw[dashed] (uj) -- (\i, 4); 
		\draw[dashed] (uj) -- (\i+1, -2); 
		\draw[dashed] (uj) -- (\i,   -2); 
		\draw[dashed] (uj) -- (\i-1, -2); 
		\node[point2] at (\i+3, -2) [label=right:$\bminus $] {}; 
		\node[point2] at (\i+3, 4) [label=right:$\tminus$] {};
		\node[point2] at (\i, 4) [label=left:$h_j$] {};
	}
	
	\draw [-stealth](-2,1) -- (-1,1);
	
	\begin{scope}[shift={(-1,0)}]	
	
	\draw (1,1) -- (8,1);
	\foreach \i in {1,2,3  }{
		\node[point] (x) at (2*\i,1) [label=above:$\xx{j}{\i}$] {};
		\node[point] (y) at (2*\i-1,1) [label=above:$\xy{j}{\i}$] {};
		\node[point] (z) at (2*\i,0) [label={[shift={(0.35,-0.6)}]:$\xz{j}{\i}$}] {};
		\node[point] (w) at (2*\i-1,0) [label=below:$\xw{j}{\i}$] {};
		\draw (x) -- (z) -- (w);
		\draw[dashed] (z) -- (2*\i,-2);
	}
	\foreach \i in {4  }{
		\node[point] (x) at (2*\i,1) [label=above:$\xx{j}{\i}$] {};
		\node[point] (y) at (2*\i-1,1) [label=above:$\xy{j}{\i}$] {};
		\node[point] (z) at (2*\i,2) [label=right:$\xz{j}{\i}$] {};
		\node[point] (w) at (2*\i-1,2) [label=below:$\xw{j}{\i}$] {};
		\draw (x) -- (z) -- (w);
		\draw[dashed] (z) -- (2*\i,4);
	}
	\node[point] (uj) at  (9,1) [label=below:$\ua{j}{5}$] {};
	\node[point] (ujplus) at  (10,2) [label=above:$u_j^+$] {};
	\node[point] (ujminus) at  (10,0) [label=above:$u_j^-$] {};
	
	\draw (8,1) -- (uj) -- (ujplus); 
	\draw (uj) -- (ujminus); 
	\draw[dashed] (ujplus) -- (12,4);
	\draw[dashed] (ujminus) -- (12,-2);
	
	\node[point2] at (12, -2) [label=right:$\bstar $] {}; 
	\node[point2] at (12, 4) [label=right:$\tstar$] {};
	\node[point2] at (8, 4) [label=right:$\cstar{j}$] {};

	\end{scope}
	
	\end{tikzpicture}
	\caption{$\plus{j}u_j\minus{j}$ $\rightarrow$ $\mg{j}$}
	\label{transformationuj}
\end{figure}
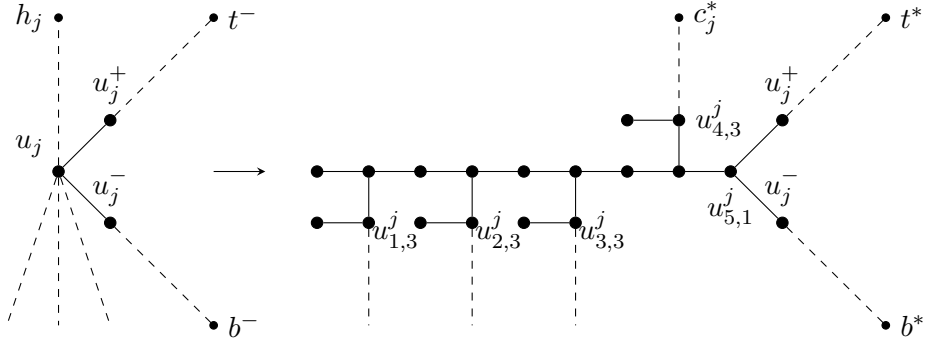

\subsection{\texorpdfstring{Replace $Y_i$ isomorphic to $K_{f_i,f_i-1}$ with $\Ynew{i}$ isomorphic to $P_{2f_i-1}$}{Modification part 3}}\label{secynew}

\newcommand{\Wplus}[1]{W_{#1}^+}
\newcommand{\Wminus}[1]{W_{#1}^-}
\newcommand{\f}[1]{ f_{#1} }
\newcommand{\w}[2]{ w_{#1,#2} }

For each $x_i\in X$ the graph $G$ contains the subgraph $Y_i$ isomorphic to $K_{f_i,f_i-1}$, where $f_i$ is the number of triples in $\mathcal{C}$ containing $x_i$. In $\Gdthree$, we replace $Y_i$ with a path $\Ynew{i}$ isomorphic to $P_{2f_i-1}$ (see Figure \ref{transformationYi}).
Like in $G$, let $\Wplus{i}\dot{\cup}\Wminus{i}$ be the bipartition of $\Ynew{i}$, where $|\Wplus{i}|=\f{i}$ and $\Wminus{i}=\f{i}-1$. As in $G$, we label the vertices of the larger bipartition $W_i^+$ as $\{w_{i,j}: x_i\in C_j \in \mathcal{C} \}$ (arbitrary but fixed) and keep the edges between $\big\{ \down{j}{1}, \down{j}{2}, \down{j}{3}\big\}$ and $\{w_{i,j}: x_i\in C_j\}$ as defined in the previous step.
In $G$, all vertices of the partition of size $f_i$ of $Y_i$ are adjacent to $\bminus$.
In $\Gdthree$, only the end-vertices of the path $\Ynew{i}$ are adjacent to $\bstar$ to receive degree at most three for all vertices of the path $\Ynew{i}$ in $\Gdthree$. 
For $f_i\ge 2$, let $w_i^*$ and $w_i^{**}$ be those end-vertices of the path $\Ynew{i}$. Note, if $f_i=1$, the unique vertex of $\Ynew{i}$, say $w_i^*$, is adjacent to $\bstar$.

\begin{figure}[htb!]
	\centering
	\begin{tikzpicture}[scale=0.575]
	
	\tikzstyle{point}=[draw,circle,inner sep=0.cm, minimum size=1.5mm, fill=black]

	\foreach \i in {-9,-7,-5,-3}{
	\foreach \j in {-10,-8,-6,-4,-2}{
		\node[point] (a) at (\i,0) [label=above:] {};
		\node[point] (a) at (\j,1) [label=above:] {};
		\draw (\i,0) -- (\j,1);
		\draw[dashed] (\j,1) -- (\j,3);
		\draw[dashed] (\j,1) -- (-2,-2);
	}	
	\node[point] (a) at (-2,-2) [label=right:$\bminus$] {};
}

	\draw [-stealth](-1,0.5) -- (0,0.5);
	
	\foreach \i in {1,3,5,7,9  }{
		\node[point] at (\i,1) [label=above:] {};	
		\draw[dashed] (\i,1) -- (\i,3);
}
	\foreach \i in {2,4,6,8  }{
	\node[point] at (\i,0) [label=above:] {};	
	\draw (\i,0) -- (\i-1,1);
	\draw (\i,0) -- (\i+1,1);
}
	\draw[dashed] (9,1) -- (9,-2);
	\draw[dashed] (1,1) -- (9,-2);
	\node[point] (a) at (9,-2) [label=right:$b^*$] {};

	\node[point] (a) at (1,1) [label=above left:$w_i^{*}$] {};
	\node[point] (a) at (9,1) [label=above right:$w_i^{**}$] {};
	
	\end{tikzpicture}
	\caption{Example for $f_i=5$: $Y_i$ in $G$ $\rightarrow$ $\Ynew{i}$ in $\Gdthree$ }
	\label{transformationYi}
\end{figure}
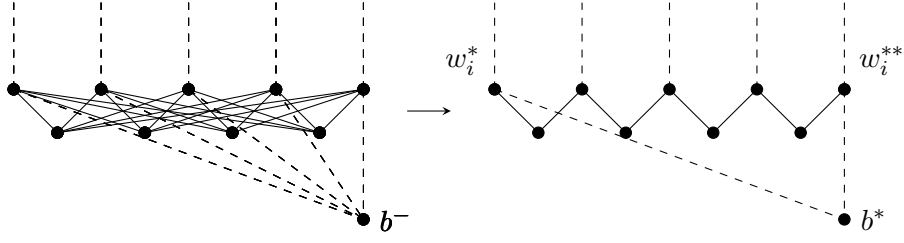

\subsection{\texorpdfstring{Replace $\tstar\tstars$ and $\bstar\bstars$ with $\mrt$ and $\mrb$, respectively}{Modification part 4}}\label{secmrt}
\newcommand{\mrta}[1]{ t_{#1,3} }
\newcommand{\mrtb}[1]{ t_{#1,4} }
\newcommand{\mrtc}[1]{ t_{#1,2} }
\newcommand{\mrtd}[1]{ t_{#1,1} }

\newcommand{\mrba}[1]{ b_{#1,3} }
\newcommand{\mrbb}[1]{ b_{#1,4} }
\newcommand{\mrbc}[1]{ b_{#1,2} }
\newcommand{\mrbd}[1]{ b_{#1,1} }
Now $\tstar$ and $\bstar$ are the only vertices of $\Gdthree$ with degree larger three. Besides $\tstars$, the vertex $\tstar$ is adjacent to all vertices in $\Aa$ and $u_j^+$ for each $C_j\in \mathcal{C}$ in $\Gdthree$, implying $\Qt \coloneqq \text{d}_{\Gdthree}(\tstar)-1=q(\Cabs-1)+\Cabs$. 
In $\Gdthree$, we replace the edge $\tstar\tstars$ with the subgraph $\mrt=\ccc{\ttt}{\Qt}$ (Definition \ref{definitionccc}, Figure \ref{transformationtandb}) and add $\Qt$
many edges such that each vertex in 
$\{  \mrta{k}: k\in [\Qt] \}$
is adjacent to exactly one vertex in $A \cup  \{ \plus{j}: C_j\in \mathcal{C}\}$ and vice versa.\\

Analogously for $\bstar$. 
Besides $\bstars$, the vertex $\bstar$ is adjacent to $u_j^-$ for each $C_j\in \mathcal{C}$ and to the end-vertices $w^*_i,w^{**}_i$ of $\Ynew{i}$  (one end-vertex $w^*_i$ if $f_i=1$) 
for each $x_i\in X$ in $\Gdthree$, implying
$\Qb\coloneqq \text{d}_{\Gdthree}(\bstar)-1=\Cabs+ \sum_{x_i \in X} \min\{ f_i,2 \}$. In $\Gdthree$, we replace the edge $\bstar\bstars$ with the subgraph $\mrb=\ccc{\bbb}{\Qb}$ (Definition \ref{definitionccc}) and add $\Qb$
many edges
such that each vertex in $\{ w_i^* : x_i\in X   \}\cup \{ w_i^{**}: x_i\in X, f_i\ge 2 \}\cup \{  \minus{j}: C_j\in \mathcal{C}  \}$ is adjacent to exactly one vertex in $\{\mrba{1},\ldots, \mrba{\Qb}   \}$ and vice versa.

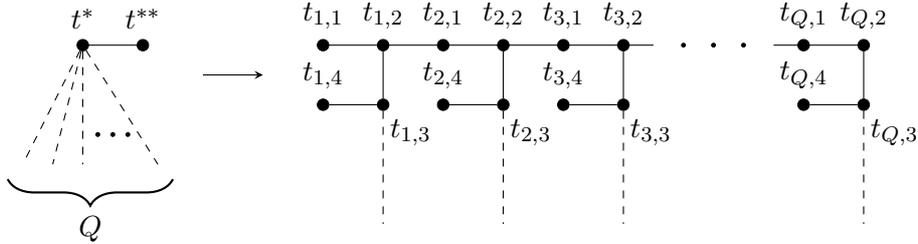
\begin{figure}[htb!]
	\centering
	\begin{tikzpicture}[scale=0.79]
	
	\tikzstyle{point}=[draw,circle,inner sep=0.cm, minimum size=1.5mm, fill=black]
	
	\tikzstyle{point2}=[draw,circle,inner sep=0.cm, minimum size=0.5mm, fill=black]
	
	\foreach \i in {-4}{
	\draw (\i,1) -- (\i+1,1);
	\node[point] (t) at (\i,1) [label=above:$t^*$] {};
	\node[point] at (\i+1,1) [label=above:$t^{**}$] {};
	\draw[dashed] (t) -- (\i-1,-1);
	\draw[dashed] (t) -- (\i-0.5,-1);
	\draw[dashed] (t) -- (\i,-1);
	
	\foreach \j in {0.25,0.5,0.75 }{
		\node[point2] at (\i+\j,-.5) [label=right:] {};
	}
	\draw[dashed] (t) -- (\i+1.25,-1);
}

\draw[thick,black,decorate,decoration={brace,amplitude=10}] (-2.5,-1.25) -- (-5.25,-1.25) node[midway, below,yshift=-10]{$\Qt$};	
	
	\draw [-stealth](-2,0.5) -- (-1,0.5);

	\begin{scope}[shift={(-1,0)}]	
	
	\draw (1,1) -- (6.5,1);
	\foreach \i in {1,2,3  }{
		\node[point] (x) at (2*\i,1) [label=above:$\mrtc{\i}$] {};
		\node[point] (y) at (2*\i-1,1) [label=above:$\mrtd{\i}$] {};
		\node[point] (z) at (2*\i,0) [label=below right:$\hspace{-1mm}\mrta{\i}$] {};
		\node[point] (w) at (2*\i-1,0) [label=above:$\mrtb{\i}$] {};
		\draw (x) -- (z) -- (w);
		\draw[dashed] (z) -- (2*\i,-2);
	}
	
	\foreach \i in {7,7.5,8 }{
		\node[point2] at (\i,1) [label=right:] {};
	}
	
	\draw (8.5,1) -- (10,1);
	\foreach \i in {5  }{
		\node[point] (x) at (2*\i,1) [label=above:$\mrtc{\Qt}$] {};
		\node[point] (y) at (2*\i-1,1) [label=above:$\mrtd{\Qt}$] {};
		\node[point] (z) at (2*\i,0) [label=below right:$\hspace{-1mm}\mrta{\Qt}$] {};
		\node[point] (w) at (2*\i-1,0) [label=above:$\mrtb{\Qt}$] {};
		\draw (x) -- (z) -- (w);
		\draw[dashed] (z) -- (2*\i,-2);
	}

\end{scope}
	
	\end{tikzpicture}
	\caption{Replace the edge $\tstar\tstars$ in $\Gdthree$ with $\mrt$}
	\label{transformationtandb}
\end{figure}
Now $\Gdthree$ has maximum degree three. This concludes the construction of $\Gdthree$.

\subsection{\texorpdfstring{Compact definition of $\Gdthree$}{Compact description of reduction}}

Using the subgraphs defined in the previous four subsections and the corresponding labeling of the vertices, we give a short definition of $\Gdthree$, summarizing the previous subsections.\\

Given the instance $(\mathcal{C}, X)$ of { \sc Exact Cover By 3-Sets} with $|\mathcal{C}|\ge q\ge  2$, the graph $\Gdthree$ is constructed as follows: 
	
\begin{enumerate}[label=(\Roman*), leftmargin=*, start=1]
	\item Generate the subgraph $\btmn{q}{\Cabs}$ (see \ref{secbm}).

	\item For every $C_i \in \mathcal{C}$, generate the subgraph $\mg{j}$ (see \ref{secmg}) and add an edge between the vertices $\down{j}{4}$ and $\cstar{j}$.
	
	\item For each $x_i \in X$, generate the subgraph $\Ynew{i}$ isomorphic to $P_{2\f{i}-1}$, where $f_i$ is the number of triples in $\mathcal{C}$ that contains the element $x_i$ (see \ref{secynew}). Add edges between the sets $\{ \down{j}{1}, \down{j}{2}, \down{j}{3}\}$ and $\{w_{i,j}: x_i\in C_j\}$ such that each vertex in $\{ \down{j}{1}, \down{j}{2}, \down{j}{3}\}$ is adjacent to exactly one vertex in $\{w_{i,j}: x_i\in C_j\}$ and vice versa.\label{cd3}
	
	\item Add the subgraphs $\mrt =\ccc{t}{\Qt}$ and $\mrb=\ccc{b}{\Qb}$, where $\Qt=q(\Cabs-1)+\Cabs$ and $\Qb=\Cabs+ \sum_{x_i \in X} \min\{ f_i,2 \}$ (see \ref{secmrt}).\label{cd4}
		Add edges such that each vertex in $A \cup \{\plus{1},\ldots, \plus{|\mathcal{C}|} \}$ is adjacent to exactly one vertex in $\{\mrta{1},\ldots, \mrta{\Qt}\}$ and vice versa. 
	Add edges such that each vertex in $\{ w_i^* : x_i\in X   \}\cup \{ w_i^{**}: x_i\in X, f_i\ge 2 \}\cup \{  \minus{j}: C_j\in \mathcal{C}  \}$ is adjacent to exactly one vertex in $\{\mrba{1},\ldots, \mrba{\Qb}   \}$ and vice versa.

\end{enumerate}
Note that there are many different ways to add the edges in \ref{cd3} and \ref{cd4}. The specific way is not relevant in the proof and, for simplicity, we do not name the edges specifically.\\

For the instance $(\mathcal{C},X)$ with $X = \{x_1,\ldots,x_6\}$ and $\mathcal{C} = \{ \{x_1,x_2,x_3\}, \{ x_1,x_2,x_4 \}, \{ x_1,x_2,x_5 \}, \\ \{ x_4,x_5,x_6 \} \}$ from the last chapter, one possible graph $\Gdthree$ can be seen in Figure \ref{figall}.

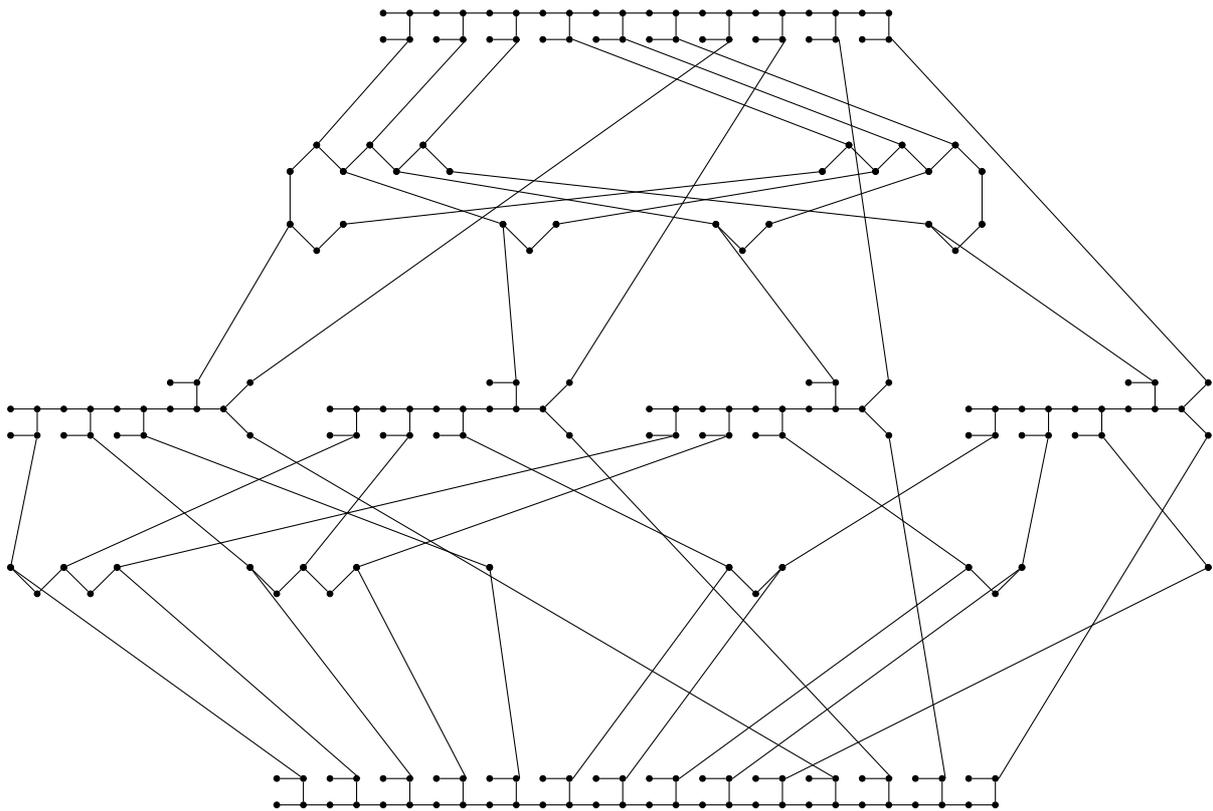
\begin{figure}[htb!]
	\centering
	\begin{tikzpicture} [scale=0.35] 	\tikzstyle{point}=[draw,circle,inner sep=0.cm, minimum size=0.75mm, fill=black]
	\tikzstyle{point2}=[draw,circle,inner sep=0.cm, minimum size=0.5mm, fill=black]

	\draw (1,1) -- (20,1);
	\foreach \i in {1,2,3,4,5,6,7,8,9,10}{
		\node[point] (x) at (2*\i,1) [label=above:] {};
		\node[point] (y) at (2*\i-1,1) [label=above:] {};
		\node[point] (t\i) at (2*\i,0) [label=below right:] {};
		\node[point] (w) at (2*\i-1,0) [label=below:] {};
		\draw (x) -- (t\i) -- (w);
	}

\begin{scope}[shift={(-4,-30)}]

	\draw (1,1) -- (28,1);
	\foreach \i in {1,2,3,4,...,14}{
		\node[point] (x) at (2*\i,1) [label=above:] {};
		\node[point] (y) at (2*\i-1,1) [label=above:] {};
		\node[point] (b\i) at (2*\i,2) [label=below right:] {};
		\node[point] (w) at (2*\i-1,2) [label=below:] {};
		\draw (x) -- (b\i) -- (w);
	}

\end{scope}

\begin{scope}[shift={(-2.5,-5)}]

\foreach \j in {0,1,2,3}{
	
	\begin{scope}[shift={(8*\j,-2)}]
 
	\node[point] (w) at (0,0) [label=above:] {};
	\foreach \i in {0}{
		\node[point] (h\j\i) at (0,0) [label=above:] {};
	}
	
	\foreach \i in {1}{
		\node[point] (ww) at (2*\i-1,-1) [label=above:] {};
		\draw (w) -- (ww);
		\node[point] (w) at (2*\i,0) [label=above:] {};
		\node[point] (h\j\i) at (2*\i,0) [label=above:] {};
		\draw (w) -- (ww);
	}
	
	\end{scope}
}

\foreach \j in {0,1}{
	
	\begin{scope}[shift={(20*\j,0)}]
	
	\node[point] (w) at (0,0) [label=above:] {};
	\foreach \i in {0}{
		\node[point] (a\j\i) at (0,0) [label=above:] {};
	}
	\foreach \i in {1,...,3}{
		\node[point] (ww) at (2*\i-1,1) [label=above:] {};
		\node[point] (bb\j\i) at (2*\i-1,1) [label=above:] {};
		\draw (w) -- (ww);
		\node[point] (w) at (2*\i,0) [label=above:] {};
		\node[point] (a\j\i) at (2*\i,0) [label=above:] {};
		\draw (w) -- (ww);
	}
	
	\end{scope}
}

\foreach \i in {0,...,3}{
	\foreach \j in {0,...,1}{
		\draw (h\i\j) -- (a\j\i);
	}
}

\end{scope}

\edef\k{1}
\foreach \i in {0,1}{
	\foreach \j in {1,2,3}{
		\draw (bb\i\j) -- (t\k);
		\pgfmathparse{\k+1}
		\xdef\k{\pgfmathresult}
	}
}

\edef\kb{11}

\foreach \l in {0,1,2,3}{

\begin{scope}[shift={(-14+12*\l,-15)}]

	\draw (1,1) -- (8,1);
\foreach \i in {1,2,3}{
	\node[point] (x) at (2*\i,1) [label=above:] {};
	\node[point] (y) at (2*\i-1,1) [label=above:] {};
	\node[point] (u\l\i) at (2*\i,0) [label=below right:] {};
	\node[point] (w) at (2*\i-1,0) [label=below:] {};
	\draw (x) -- (u\l\i) -- (w);
}
\foreach \i in {4  }{
	\node[point] (x) at (2*\i,1) [label=above:] {};
	\node[point] (y) at (2*\i-1,1) [label=above:] {};
	\node[point] (z) at (2*\i,2) [label=above right:] {};
	\node[point] (w) at (2*\i-1,2) [label=below:] {};
	\draw (x) -- (z) -- (w);
	\foreach \j in {0}{
		\draw (z) -- (h\l\j);
	}
}
\node[point] (uj) at  (9,1) [label=above:] {};
\node[point] (ujplus) at  (10,2) [label=above:] {};
\node[point] (ujminus) at  (10,0) [label=above:] {};
\draw (8,1) -- (uj) -- (ujplus) -- (t\k); 
\pgfmathparse{\k+1}
\xdef\k{\pgfmathresult}

\draw (uj) -- (ujminus) -- (b\kb); 
\pgfmathparse{\kb+1}
\xdef\kb{\pgfmathresult}

\end{scope}

}

\def\myarray{{3,3,1,2,2,1}}
\edef\m{0}

\edef\kbb{1}

\foreach \j in {0,1,2,3,4,5}{
	\pgfmathparse{int(\myarray[\m]-1)}
	\let\result\pgfmathresult
	\pgfmathparse{\m+1}
	\xdef\m{\pgfmathresult}
	
	\begin{scope}[shift={(-10-3+9*\j,-20)}]

	\node[point] (w) at (0,0) [label=above:] {};
	\foreach \i in {0}{
		\node[point] (w\j\i) at (0,0) [label=above:] {};
	}
	\draw (w) -- (b\kbb);
	\pgfmathparse{\kbb+1}
	\xdef\kbb{\pgfmathresult}
	\if\result0\else
		\foreach \i in {1,...,\result}{
			\node[point] (ww) at (2*\i-1,-1) [label=above:] {};
			\draw (w) -- (ww);
			\node[point] (w) at (2*\i,0) [label=above:] {};
				\node[point] (w\j\i) at (2*\i,0) [label=above:] {};
			\draw (w) -- (ww);
		}
		\draw (w) -- (b\kbb);
		\pgfmathparse{\kbb+1}
		\xdef\kbb{\pgfmathresult}
	\fi
\end{scope}
}

 \def\myarr{{0,1,2,0,1,3,0,1,4,3,4,5}}
\def\myarrt{{0,0,0,1,1,0,2,2,0,1,1,0}}
\edef\s{0}

\foreach \i in {0,1,2,3}{
	\foreach \j in {1,2,3}{
		\pgfmathparse{\myarr[\s]}
		\let\x\pgfmathresult
		\pgfmathparse{\myarrt[\s]}
		\let\y\pgfmathresult
		\draw (u\i\j) -- (w\x\y);
		\pgfmathparse{\s+1}
		\xdef\s{\pgfmathresult}
	}
}
	\end{tikzpicture}
	\caption{The graph
 $\Gdthree$ for $X = \{x_1,\ldots,x_6\}$ and $\mathcal{C} = \{ \{x_1,x_2,x_3\}, \{ x_1,x_2,x_4 \}, \{ x_1,x_2,x_5 \},\\ \{ x_4,x_5,x_6 \} \}$}
	\label{figall}
\end{figure}

\subsection{Proof of Theorem \ref{theoremnunuddelta3}}

The proof consists of three parts. First, we determine the matching number of $\Gdthree$ and show that $\nu_d(\Gdthree)=\nu(\Gdthree)$ if $(\mathcal{C},X)$ is a \YES-instance. Second, we show that, for a disconnected matching $M$ of size $\nu(\Gdthree)$, $\Gdthree[M]$ has exactly two components. Third, we show that $(\mathcal{C},X)$ is a \YES-instance if $\nu_d(\Gdthree)=\nu(\Gdthree)$ and conclude the proof of Theorem \ref{theoremnunuddelta3}.\\

By construction, $\Gdthree$ has maximum degree three. Also, $\Gdthree$ is indeed bipartite with the bipartition $V_1\dot\cup V_2$ where 
\begin{align*}
V_1=\Bb\cup \Dd \cup \bigcup_{C_j\in \mathcal{C},k\in [5]} \big\{ \ua{j}{k} \big\} \bigcup_{C_j\in \mathcal{C},k\in [4]} \big\{ \down{j}{k}  \big\}\cup \bigcup_{x_i\in X} \Wminus{i} 
\cup \bigcup_{k\in [\Qt]} \left\{ \mrtd{k},\mrta{k}\right\} \cup \bigcup_{k\in [\Qb]}\left\{\mrbd{k},\mrba{k} \right\}
\end{align*}
and $V_2$ contains all the remaining vertices.
	
\begin{lemma}\label{matchingHqC} The graph $\btmn{q}{\Cabs}$ has matching number $\nu(\btmn{q}{\Cabs})=\Cabs(2q-1)$ and for any subset $S$ of $\{\cstar{1},\cstar{2},\ldots, \cstar{\Cabs}\}$ with cardinality $\Cabs-q$, there exists a maximum matching of $\btmn{q}{\Cabs}$ saturating none of the vertices in $S$ but all other vertices.
\end{lemma}	

\begin{proof} $\btmn{q}{\Cabs}$ is bipartite with bipartitions of cardinality $\Cabs(2q-1)$ and $(2\Cabs-1)q$, implying $\nu(\btmn{q}{\Cabs})\le \Cabs(2q-1)$.\\
Let $S\subseteq \{\cstar{1},\cstar{2},\ldots, \cstar{\Cabs}   \}$ of cardinality $\Cabs-q$. We give a matching of size $\Cabs(2q-1)$ not saturating any vertex in $S$ (implying both statements).

Let $J\subseteq [\Cabs ]$ such that $S=\{ \cstar{j} : j\in J \}$. By construction, there is an edge between any path $Q_k, k\in [q]$ and $R_j, j\in \Cabs$. Hence there exist $q$ many edges such that any path $Q_k, k\in [q]$ and any path $R_j, j\in \Cabs \setminus J$ contains exactly one of the 
end-vertices
of those edges. Those edges define a matching $M'$. Now we delete all edges between the sets $B$ and $C$. We delete all $2q$ many vertices, that are saturated by $M'$ and we delete the vertices $\cstar{j}$ for all $j\in J$, since we are not allowed to saturate those vertices. Note that we removed exactly one vertex from every path $Q_k, R_j$. The remaining graph with $2\big[q(|\mathcal{C}|-1)+(q-1) |\mathcal{C}|\big]$ many vertices is the union of odd paths and hence has a perfect matching $M''$. The union of $M'$ and $M''$ yields the desired matching of $\btmn{q}{\Cabs}$ (see example in Figure \ref{H35matching}).
\end{proof}

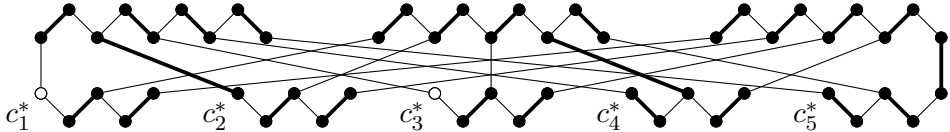
\begin{figure}[htb!]
	\centering
	\begin{tikzpicture}[scale=0.37]
	
	\tikzstyle{point}=[draw,circle,inner sep=0.cm, minimum size=1.5mm, fill=black]
	\tikzstyle{point2}=[draw,circle,inner sep=0.cm, minimum size=1.5mm, fill=white]
	\tikzstyle{matching}=[line width=0.5mm]
	
	\foreach \j in {0,1,2,3,4}{
		
		\begin{scope}[shift={(7*\j,-2)}]
		
		\node[point] (w) at (0,0) [label=above:] {};
		\foreach \i in {0}{
			\node[point] (h\j\i) at (0,0) [label=above:] {};
		}
		
		\foreach \i in {1,...,2}{
			\node[point] (ww) at (2*\i-1,-1) [label=above:] {};
			\node[point] (d\j\i) at (2*\i-1,-1) [label=above:] {};
			\draw (w) -- (ww);
			\node[point] (w) at (2*\i,0) [label=above:] {};
			\node[point] (h\j\i) at (2*\i,0) [label=above:] {};
			\draw (w) -- (ww);
		}
		
		\end{scope}
	}
	
	\foreach \j in {0,1,2}{
		
		\begin{scope}[shift={(12*\j,0)}]
		
		\node[point] (w) at (0,0) [label=above:] {};
		\foreach \i in {0}{
			\node[point] (a\j\i) at (0,0) [label=above:] {};
		}
		\foreach \i in {1,...,4}{
			\node[point] (ww) at (2*\i-1,1) [label=above:] {};
			\node[point] (b\j\i) at (2*\i-1,1) [label=above:] {};
			\draw (w) -- (ww);
			\node[point] (w) at (2*\i,0) [label=above:] {};
			\node[point] (a\j\i) at (2*\i,0) [label=above:] {};
			\draw (w) -- (ww);
		}
		
		\end{scope}
	}
	
	\foreach \i in {0,...,4}{
		\foreach \j in {0,...,2}{
			\draw (h\i\j) -- (a\j\i);
		}
	}
	
	\foreach \j in {0,2}{
		\pgfmathparse{int(\j+1)}
		\let\jplus\pgfmathresult
		\foreach \i in {0}{
			\node[point2] (a) at (h\j\i) [label={[shift={(-0.3,-0.7)}]:$\cstar{\jplus}$}] {};
		}	
	}
	
	\foreach \j in {1,3,4}{
		\pgfmathparse{int(\j+1)}
		\let\jplus\pgfmathresult
		\foreach \i in {0}{
			\node[point] (a) at (h\j\i) [label={[shift={(-0.3,-0.7)}]:$\cstar{\jplus}$}] {};
		}	
	}

	\foreach \i in {1}{
		\foreach \j in {0}{
			\draw[matching] (h\i\j.center) -- (a\j\i.center);
		}
	}	
	
	\foreach \i in {3}{
		\foreach \j in {1}{
			\draw[matching] (h\i\j.center) -- (a\j\i.center);
		}
	}
	
	\foreach \i in {4}{
		\foreach \j in {2}{
			\draw[matching] (h\i\j.center) -- (a\j\i.center);
		}
	}

	\foreach \i in {0}
	{
		\foreach \j in {2,3,4}{
			\draw[matching] (a\i\j.center)-- (b\i\j.center);		
		}
	}

	\foreach \i in {1}
	{
		\foreach \j in {4}{
			\draw[matching] (a\i\j.center)-- (b\i\j.center);		
		}
	}
	
	\foreach \i in {0}
	{
		\foreach \j in {0}{
			\pgfmathparse{int(\j+1)}
			\let\jplus\pgfmathresult
			\draw[matching] (a\i\j.center)-- (b\i\jplus.center);		
		}
	}
	
	\foreach \i in {1}
	{
		\foreach \j in {0,1,2}{
			\pgfmathparse{int(\j+1)}
			\let\jplus\pgfmathresult
			\draw[matching] (a\i\j.center)-- (b\i\jplus.center);		
		}
	}
	
	\foreach \i in {2}
	{
		\foreach \j in {0,1,2,3}{
			\pgfmathparse{int(\j+1)}
			\let\jplus\pgfmathresult
			\draw[matching] (a\i\j.center)-- (b\i\jplus.center);		
		}
	}

	\foreach \i in {0,1,2}
	{
		\foreach \j in {1,2}{
			\draw[matching] (h\i\j.center)-- (d\i\j.center);		
		}
	}
	
	\foreach \i in {3}
	{
		\foreach \j in {2}{
			\draw[matching] (h\i\j.center)-- (d\i\j.center);		
		}
	}
	
	\foreach \i in {3}
	{
		\foreach \j in {0}{
			\pgfmathparse{int(\j+1)}
			\let\jplus\pgfmathresult
			\draw[matching] (h\i\j.center)-- (d\i\jplus.center);		
		}
	}
	
	\foreach \i in {4}
	{
		\foreach \j in {0,1}{
			\pgfmathparse{int(\j+1)}
			\let\jplus\pgfmathresult
			\draw[matching] (h\i\j.center)-- (d\i\jplus.center);		
		}
	}
	
	\end{tikzpicture}
	\caption{Example for $\btmn{3}{5}$ and $S=\{\cstar{1}, \cstar{3}  \}$}
	\label{H35matching}
\end{figure}

\begin{lemma}\label{lemmaifYESthennuisnud} The graph $\Gdthree$ has matching number $\nu(\Gdthree)=|V_1|$ and if the instance $(\mathcal{C},X)$ of {\sc Exact Cover By 3-Sets} is a YES-instance, then $\Gdthree$ has a disconnected matching of size $\nu(\Gdthree)=|V_1|$.\\
\end{lemma}

\begin{proof}
	Since $\Gdthree$ has the bipartition $V_1\dot\cup V_2$, it holds that $\nu(\Gdthree)\le |V_1|$.
	Let $\mathcal{C}'\subseteq \mathcal{C}$ with $|\mathcal{C}'|=q$.

	\begin{enumerate}
		\item By Lemma \ref{matchingHqC}, $\btmn{q}{|\mathcal{C}|}$ has a matching of cardinality $\Cabs(2q-1)$, that does not saturate $\cstar{j}$ for every $C_j\in \mathcal{C} \setminus \mathcal{C}'$ but all other vertices.
		
		\item For every $C_j\in \mathcal{C}$, we choose a perfect matching of $\mg{j}-\minus{j}$ if $C_j\in \mathcal{C}'$ and of $\mg{j}-\plus{j}$         otherwise.
		
		\item For every $x_i\in X$, the subgraph $\Ynew{i}$ isomorphic to $P_{2f_i-1}$ has a matching of size $\f{i}-1$ saturating its smaller bipartiton $\Wminus{i}$.\label{mat3}
		
		\item $\mrt$ and $\mrb$ have perfect matching of size $2\Qt$ and $2\Qb$, respectively.
	\end{enumerate}
		    The union of the matchings of those subgraphs yields a matching of $\Gdthree$ that saturates all vertices in $V_1$, implying $\nu(\Gdthree)\ge|V_1|$ and hence $\nu(\Gdthree)=|V_1|$. \\
	
	Assume $(\mathcal{C},X)$ is a YES-instance and let $\mathcal{C}'\subseteq \mathcal{C}$ with $|\mathcal{C}'|=q$ be an exact cover of $X$. We modify the matching defined above by replacing \ref{mat3} with \ref{mat3prime}
        and call this matching $M$.
	
	\begin{enumerate}[label=3'.,ref=3']
		\item\label{mat3prime}
                        Every $x_i \in X$ is contained in exactly one triple $C_j \in \mathcal{C}'$. The subgraph $\Ynew{i} - \w{i}{j} $ 
		consists of one or two paths of odd length
        		and hence has a perfect matching of size $\f{i}-1$.
  	\end{enumerate}
	
The described matching $M$ saturated all vertices of $\Gdthree$ except
	\begin{itemize}
		\item $\cstar{j}$ and $\plus{j}$ for all $C_j\in \mathcal{C}\setminus \mathcal{C}'$,
		\item $\minus{j}$ for all $C_j\in \mathcal{C}'$ and
		\item $\w{i}{j}$ for all $C_j\in \mathcal{C}'$ and all $x_i\in C_j$.
	\end{itemize}

	The matching $M$ saturates all vertices in $V_1$, implying it is a maximum matching.
	To show that $M$ is disconnected,
		let $U$ be the set of all saturated vertices of $\mrt$, $\btmn{q}{\Cabs}$ and $\mg{j}$ for all $C_j\in \mathcal{C}'$, hence
	\begin{align*}
	U = V(\mrt) \cup V\left(\btmn{q}{\Cabs}\right) \setminus  \left\{ \cstar{j}: C_j \in \mathcal{C}\setminus \mathcal{C}' \right\}  \cup \bigcup_{C_j \in \mathcal{C}'} V\left(\mg{j}\right) \setminus \{\minus{j}\}.
	\end{align*}
	
	Clearly $\Gdthree[U]\not=\Gdthree[M]$.
     $N_{\Gdthree}[U]\setminus U$ is the union of the following sets:
	\begin{itemize}
		\item $N_\Gdthree(V(\mrt))\setminus U=\{ \plus{j}: C_j\in \mathcal{C}'  \}$
		\item $N_\Gdthree\Big(V\big(\btmn{q}{\Cabs}\big) \setminus \big\{ \cstar{j}: C_j \in \mathcal{C}\setminus \mathcal{C}' \big\} \Big) \setminus U=\{ \cstar{j}: C_j \in \mathcal{C}\setminus \mathcal{C}' \}$
		\item $N_\Gdthree\left (\bigcup_{C_j \in \mathcal{C}'} V\big(\mg{j}\big)\setminus \big\{ \minus{j} \big\} \right)\setminus U=\left \{  \w{i}{j}: x_i\in C_j\in \mathcal{C}' \right \} \cup \bigcup_{C_j \in \mathcal{C}'} \minus{j}$
	\end{itemize}
	None of those vertices is saturated, implying the matching is disconnected. This concludes the proof of Lemma \ref{lemmaifYESthennuisnud}.
\end{proof}

Let $M$ be a disconnected matching of $\Gdthree$ of size $\nu(F)=|V_1|$. The vertices $\mrtd{1}\in V(\mrt)$ and $\mrbd{1}\in V(\mrb)$ are in $V_1$ and hence saturated by $M$. 

\begin{definition}
Let $T$ contain all vertices that are in the same component of $\Gdthree[M]$ as $\mrtd{1}$.
Let $B$ contain all vertices that are in the same component of $\Gdthree[M]$ as $\mrbd{1}$.
\end{definition}	 

Since $M$ is disconnected, the following lemma implies that $\Gdthree[M]$ has exactly the components $\Gdthree[T]$ and $\Gdthree[B]$.

\begin{lemma}\label{satinTorB}$ $\\
\vspace{-5mm}
\begin{enumerate}
	\item All saturated vertices of $\mrt$ are in $T$.
	\item All saturated vertices of $\mrb$ are in $B$.
	\item All saturated vertices of $\btmn{q}{|\mathcal{C}|}$ are in $T$.
	\item All saturated vertices of $\mg{j}$ are in $T$, if $\plus{j}$ is saturated, else in $B$ for every $C_j\in \mathcal{C}$.
        						\item All saturated vertices of $\Ynew{i}$ are in $B$ for every $x_i\in X$.
\end{enumerate}
\end{lemma}
		
\begin{proof}$ $\\
\vspace{-5mm}
\begin{enumerate}	
\item\label{mrtinT} The vertices $\mrtd{k},\mrta{k}$ are in $V_1$ for every $k\in [\Qt]$ and hence saturated by $M$. Since $\mrtd{1}$ has only the neighbor $\mrtc{1}$, $M$ contains the edge $\mrtd{1}\mrtc{1}$. The only edge left to saturate the vertex $\mrtd{2}$ is the edge $\mrtd{2}\mrtc{2}$ and therefore this edge is in $M$. Iteratively, all edges $\mrtd{k}\mrtc{k}$ for $k\in [\Qt]$ are in $M$ (see Figure \ref{mrtmatching}). Hence $\mrtd{1}\in T$ implies $\mrtd{k},\mrtc{k} \in T$ for every $k\in [\Qt]$.
	Since $\mrta{k}$ is saturated by $M$, $\mrtc{k}\in T$ implies $\mrta{k}\in T$. If $\mrtb{k}$ is saturated by $M$, then $M$ contains the edge $\mrtb{k}\mrta{k}$ and $\mrta{k}\in T$ implies $\mrtb{k}\in T$.
	Therefore every saturated vertex of $\mrt$ is in $T$.

	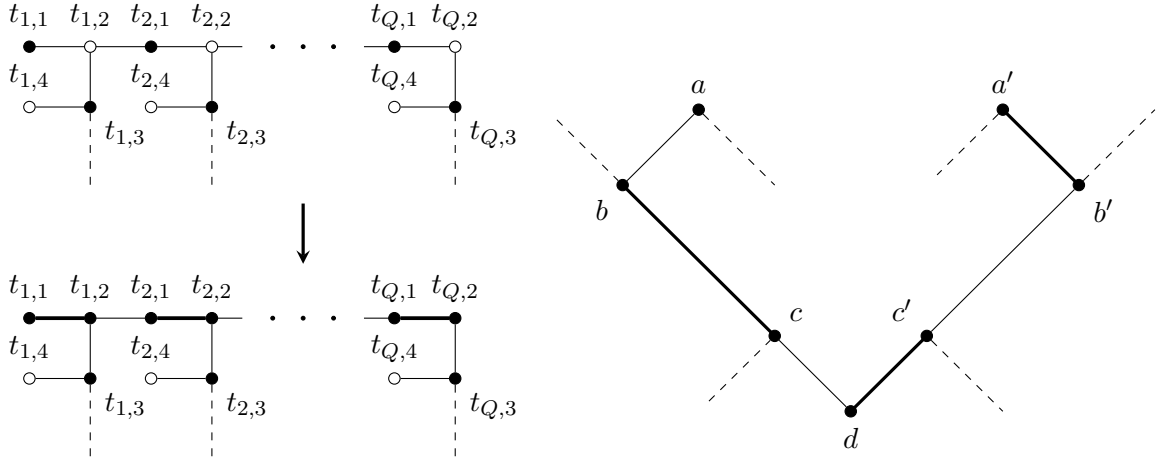
\begin{figure}[htb!]
	\centering
	\begin{tikzpicture}[xscale=0.8, yscale=0.8]
	
	\newcommand{\vx}[1]{ \mrtc{#1}}
	\newcommand{\vy}[1]{ \mrtd{#1}}
	\newcommand{\vstar}[1]{ \mrta{#1}}
	\newcommand{\vz}[1]{ \mrtb{#1}}
	
	\tikzstyle{point}=[draw,circle,inner sep=0.cm, minimum size=1.5mm, fill=black]
	\tikzstyle{pointb}=[draw,red,circle,inner sep=0.cm, minimum size=1.5mm, fill=red]
	\tikzstyle{point2}=[draw,circle,inner sep=0.cm, minimum size=0.5mm, fill=black]
\tikzstyle{point3}=[draw,circle,inner sep=0.cm, minimum size=1.5mm, fill=white]
 
	\draw (1,1) -- (4.5,1);
	\foreach \i in {1,2}{
		\node[point3] (x) at (2*\i,1) [label=above:$\vx{\i}$] {};
		\node[point] (y) at (2*\i-1,1) [label=above:$\vy{\i}$] {};
		\node[point] (z) at (2*\i,0) [label=below right:$\vstar{\i}$] {};
		\node[point3] (w) at (2*\i-1,0) [label=above:$\vz{\i}$] {};
		\draw (x) -- (z) -- (w);
		\draw[dashed] (z) -- (2*\i,-1.3);
	}
	
	\foreach \i in {5,5.5,6 }{
		\node[point2] at (\i,1) [label=right:] {};
	}
	
	\draw (6.5,1) -- (8,1);
	\foreach \i in {4  }{
		\node[point3] (x) at (2*\i,1) [label=above:$\vx{\Qt}$] {};
		\node[point] (y) at (2*\i-1,1) [label=above:$\vy{\Qt}$] {};
		\node[point] (z) at (2*\i,0) [label=below right:$\vstar{\Qt}$] {};
		\node[point3] (w) at (2*\i-1,0) [label=above:$\vz{\Qt}$] {};
		\draw (x) -- (z) -- (w);
		\draw[dashed] (z) -- (2*\i,-1.3);
	}
	
	\draw [-stealth, very thick](5.5,-1.6) -- (5.5,-2.6);
	
	\begin{scope}[shift={(0,-4.5)}]
	
	\draw (1,1) -- (4.5,1);
	\foreach \i in {1,2}{
		\node[point] (x) at (2*\i,1) [label=above:$\vx{\i}$] {};
		\node[point] (y) at (2*\i-1,1) [label=above:$\vy{\i}$] {};
		\node[point] (z) at (2*\i,0) [label=below right:$\vstar{\i}$] {};
		\node[point3] (w) at (2*\i-1,0) [label=above:$\vz{\i}$] {};
		\draw (x) -- (z) -- (w);
		\draw[line width=0.5mm] (x) -- (y);
		\draw[dashed] (z) -- (2*\i,-1.3);
	}
	
	\foreach \i in {5,5.5,6 }{
		\node[point2] at (\i,1) [label=right:] {};
	}
	
	\draw (6.5,1) -- (8,1);
	\foreach \i in {4  }{
		\node[point] (x) at (2*\i,1) [label=above:$\vx{\Qt}$] {};
		\node[point] (y) at (2*\i-1,1) [label=above:$\vy{\Qt}$] {};
		\node[point] (z) at (2*\i,0) [label=below right:$\vstar{\Qt}$] {};
		\node[point3] (w) at (2*\i-1,0) [label=above:$\vz{\Qt}$] {};
		\draw[line width=0.5mm] (x) -- (y);
		\draw (x) -- (z) -- (w);
		\draw[dashed] (z) -- (2*\i,-1.3);
	}
	\end{scope}

	\end{tikzpicture}
	\begin{tikzpicture}[scale=1]	
	\tikzstyle{point}=[draw,circle,inner sep=0.cm, minimum size=1.5mm, fill=black]
	
	\node (bb2) at (-1,3) [label=above:] {};
	\node[point] (bb3) at (1,3) [label=above:$a$] {};
	
	\draw[dashed](bb3) -- (2,2);
	
	\node[point] (ba3) at (0,2) [label=below left:$b$] {};

	\begin{scope}[shift={(0,0)}]
	
	\node[point] (h2) at (2,0) [label=above right:$c $] {};
	\node[point] (h3) at (4,0) [label=above left: $c' $] {};
	
	\node (bd1) at (1,-1) [] {};
	\draw[dashed](bd1) -- (h2);
	
	\draw[very thick] (h2.center) -- (ba3.center);
	
	\node[point] (bd2) at (3,-1) [label=below:$d$] {};
	
	\draw[dashed] (h3) -- (5,-1);
	
	\draw (ba3) -- (bb3) ;
	
	\draw[dashed] (bb2) -- (ba3);
	
	\draw (h2) -- (bd2);
	
	\draw[very thick] (bd2.center) -- (h3.center);
	
	\end{scope}

	\begin{scope}[shift={(6,0)}]
	
	\node[point] (bb2) at (-1,3) [label=above:$a' $] {};
	
	\node (ba2) at (-2,2) [label=above:] {};
	
	\draw[dashed](bb2) -- (ba2);
	\node[point] (ba3) at (0,2) [label=below right:$b'$] {};
	
	\draw[very thick] (ba3.center) -- (bb2.center);
	
	\draw[dashed](ba3) -- (1,3);	
	
	\end{scope}
	
	\draw (h3) -- (ba3);

	\end{tikzpicture}
	\caption{Illustration of the proof that all saturated vertices of $\mrt$ belong to $T$ (left) and that $\bx$ is in $T$ for every $\bx \in \Bb$ (right).}
	\label{mrtmatching}
\end{figure}

	\item Follows analogously to \ref{mrtinT} since $\mrbd{k},\mrba{k} \in V_1$ for every $k\in [\Qt]$.
	
	\item Since every vertex in $\Aa$ has a neighbor in the set $\{ \mrta{1},\ldots, \mrta{\Qt}  \}\subseteq T$, every saturated vertex of $\Aa$ is in $T$.\\
	Every vertex $\bx \in \Bb\subseteq V_1$ is saturated by an edge $e\in M$. The vertex $\bx$ has one neighbor in $\Cc$ and one or two neighbors in $\Aa$. If $e=\ax\bx$ for a vertex $\ax\in \Aa$, then $\ax \in T$ implies $\bx\in T$. Now let $e=\bx \cx$ for a vertex $\cx\in \Cc$ (see Figure \ref{mrtmatching}). (Note since $\Aa\subseteq V_2$, it is not clear whether $\bx$ has a saturated neighbor in $\Aa$ or not.)
	The vertex $\cx$ has a neighbor $\dx \in \Dd$. Let $\cx'\in \Cc$ be the other neighbor of the vertex $\dx$. 
	Since the vertex $\dx\in V_1$ must be saturated by $M$, the edge $\dx\cx'$ is in $M$. Let $\bx'$ be the neighbor of $\cx'$ in $\Bb$ ($\bx\not=\bx'$). The vertex $\bx'\in V_1$ must be saturated by $M$. $\bx'$ has the neighbor $\cx'$ and one or two neighbors in $\Aa$, implying $\ax'\bx'\in M$ for a vertex $\ax'\in \Aa$. Then $\ax'\in T$ implies $\bx',\cx',\dx,\cx,\bx\in T$. Hence $\Bb \subseteq T$.\\
 		Since any vertex in $\Cc$ has a neighbor in $\Bb\subseteq T$, any saturated vertex in $\Cc$ is in $T$. Since $\Dd\subseteq V_1$, any vertex $\dx\in \Dd$ is saturated by an edge $\dx\cx\in M$ for a vertex $\cx\in \Cc$. Then $\cx\in T$ implies $\dx \in T$.

\item Consider $\mg{j}$ for $C_j\in \mathcal{C}$ (see Figure \ref{transformationuj}). Since $\ua{j}{k},\down{j}{k}\in V_1$ for every $k\in [4]$, it follows analogously to part \ref{mrtinT}. that for every $k\in [4]$, the edge $\ua{j}{k}\ub{j}{k}$ is in $M$, $\down{j}{k}$ is saturated by $M$ and if $\downb{j}{k}$ is saturated by $M$, then $\downb{j}{k}$ is in the same component of $\Gdthree[M]$ as $\ub{j}{k}$.
To saturate $\ua{j}{5}\in V_1$, $M$ contains either the edge $\plus{j}\ua{j}{5}$ or $\minus{j}\ua{j}{5}$, implying $\plus{j}$ or $\minus{j}$ is saturated by $M$. Note that the saturated vertices of $\mg{j}$ induce a connected subgraph in $\Gdthree[M]$ and therefore are in the same component of $\Gdthree[M]$. If $\plus{j}$ is saturated by $M$, then all saturated vertices of $\mg{j}$ are in $T$. If $\minus{j}$ is saturated by $M$, then all saturated vertices of $\mg{j}$ are in $B$. 

\item Consider $\Ynew{i}$ for $x_i\in X$. First let $\f{i}\ge 2$. 
The end-vertices $w_i^*$ and $w_i^{**}$ of the path $\Ynew{i}$ are each adjacent to a vertex in $\{\mrba{k}:k\in [\Qb]\}\subseteq B$.
Since $\Wminus{i}\subseteq V_1$, all vertices of $\Wminus{i}$ are saturated by $M$. Since $N(\Wminus{i})\subseteq \Wplus{i}$ and $|\Wplus{i}|=|\Wminus{i}|+1$, at most one vertex of $\Wplus{i}$ is not saturated by $M$, yielding the following cases:

\begin{itemize}
	\item If all vertices of $\Wplus{i}$ are saturated by $M$, then $\Ynew{i}$ is a subgraph of $\Gdthree[M]$. $w_i^*$ and $w_i^{**}$ are saturated, implying $V(\Ynew{i}) \subseteq B$.	

	\item If an end-vertex of $\Ynew{i}$ is not saturated, w.l.o.g. $w_i^*$, then all vertices in $\Wplus{i}\setminus \{w_i^*\}$ are saturated by $M$ and $Y_i-w_i^*$ is a connected subgraph of $\Gdthree[M]$. $w_i^{**}$ is saturated, implying $V(\Ynew{i})\setminus \{w_i^*\}\subseteq B$.
	
    \item If a vertex $w_i'\in \Wplus{i}\setminus \{ w_i^*, w_i^{**}  \}$ is not saturated by $M$, $\Gdthree[V(\Ynew{i}) \cap V(M)]$ breaks into two components. Since one component contains  $w_i^*$ and the other component contains $w_i^{**}$, we get $V(\Ynew{i}) \setminus \{w_i'\} \subseteq B$ (see Figure \ref{yic}).
\end{itemize}
\begin{figure}[htb!]
	\centering
	\begin{tikzpicture}	
	\tikzstyle{point}=[draw,circle,inner sep=0.cm, minimum size=1.5mm, fill=black]

	\node[point] (a1) at(-1,1) [label=above:$w_i^{*}$] {};
	\node[point] (a2) at (1,1) [label=above:] {};
	\node[point] (a3) at (3,1) [label=above:$w_i'$] {};
	\node[point] (a4) at (5,1) [label=above:] {};
	\node[point] (a5) at (7,1) [label=above:$w_i^{**}$] {};
	
	\node[point] (b1) at (0,0) [label=above:] {};
	\node[point] (b2) at (2,0) [label=above:] {};
	\node[point] (b3) at (4,0) [label=above:] {};
	\node[point] (b4) at (6,0) [label=above:] {};
	
	\draw[line width=0.5mm] (a1.center) -- (b1.center);
 \draw (b1.center) -- (a2.center);
\draw[line width=0.5mm] (a2.center) --(b2.center);
	
	\draw (a5) -- (9,3);
	\draw (a1) -- (1,3);
	
	\draw[dashed] (b3) --  (a3)-- (b2);
	
	\draw[line width=0.5mm] (a5.center) -- (b4.center);
 \draw (b4.center) -- (a4.center);
 \draw[line width=0.5mm] (a4.center) --(b3.center);
	
	\end{tikzpicture}
	\caption{Illustration of $F[V(\Ynew{i}) \cap V(M)]$ when containing two components. }
	\label{yic}
\end{figure}
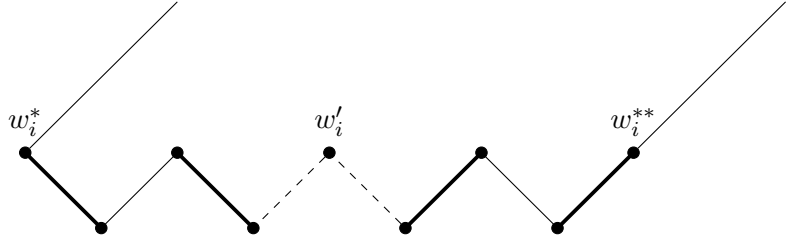
Now let $\f{i}=1$. The unique vertex $w_i^*$ of $\Ynew{i}$ is adjacent to a vertex in $\{\mrba{k}:k\in [\Qb]\}\subseteq B$. If $w_i^*$ is saturated by $M$, then $w_i^*\in B$.
\end{enumerate}
\end{proof}

\begin{lemma}\label{lemmaifnuisnudthenYES} If $\Gdthree$ has a disconnected matching $M$ of size $\nu(\Gdthree)=|V_1|$, then the instance $(\mathcal{C},X)$ of {\sc Exact Cover By 3-Sets} is a YES-instance.
\end{lemma}

\begin{proof}
Let $M$ be a disconnected matching in $\Gdthree$ of size $\nu(\Gdthree)=|V_1|$ and let $T$ and $B$ be defined as above. Therefore $\Gdthree[M]$ has exactly the components $\Gdthree[T]$ and $\Gdthree[B]$.\\

All vertices in $\Bb$ and $\Dd$ are saturated by $M$ because $ \Bb \cup \Dd \subseteq V_1$. Since $N_{\Gdthree}( \Bb\cup\Dd )=\Aa \cup \Cc$, at least $| \Bb \cup \Dd|= \Cabs (q-1)+ \Cabs q$ many vertices of $\Aa \cup\Cc$ are saturated by $M$. Since $|\Aa \cup \Cc|=\Cabs q +(\Cabs-1)q$, at most $\Cabs-q$ many vertices of $\Aa \cup \Cc$ are not saturated by $M$.
Hence at least $q$ many vertices of $\{\cstar{1},\cstar{2},\ldots, \cstar{|\mathcal{C}|}\}$ are saturated by $M$ and those saturated vertices are in $T$ by Lemma \ref{satinTorB}. Let $J^+\subseteq [ |\mathcal{C}|  ]$ be the set of indices for which $\cstar{j}$ is saturated by $M$. \\

Let $j\in J^+$. Since $\down{j}{4} \in V_1$ and hence saturated by any maximum matching, the edge $\cstar{j}\down{j}{4}$ together with $\cstar{j}\in T$ implies $\down{j}{4} \in T$ and by Lemma \ref{satinTorB} all saturated vertices of $\mg{j}$ are in $T$. Particularly $\down{j}{1},\down{j}{2},\down{j}{3}\in V_1$ are saturated by $M$ and hence $\{\down{j}{1},\down{j}{2},\down{j}{3}\} \subseteq T$ for all $j\in J^+$. Since $T\cap \bigcup_{x_i\in X} \Wplus{i}=\emptyset$, it follows that $T$ has at least $3|J^+|$ many pairwise distinct neighbors in $\bigcup_{x_i\in X} \Wplus{i}$.\\

Since $\Wplus{i}\cap V(M) \subseteq B$ for all $i\in [3q]$ by Lemma \ref{satinTorB} and $|\Wplus{i}\cap V(M)|\ge |\Wplus{i}|-1$ for all $i\in [3q]$ (see proof of Lemma \ref{satinTorB}), there are at most $|X|=3q$ many vertices in $\bigcup_{x_i\in X} \Wplus{i}$ that are not in $B$.\\

Since $N_{\Gdthree}(T) \cap B =\emptyset$, the previous two arguments imply that $T$ and hence $\bigcup_{j\in J^+} \{ \down{j}{1}, \down{j}{2}, \down{j}{3} \} $ have exactly $3q$ many neighbors in $\bigcup_{x_i\in X} \Wplus{i}$ and exactly those $3q$ vertices of $\bigcup_{x_i\in X} \Wplus{i}$ are not saturated by $M$. Consequently, $|J^+|=q$. For any $x_i\in X$ at most one vertex of $\Wplus{i}$ is not saturated by $M$, implying exactly one vertex of $\Wplus{i}$ is not saturated by $M$ for every $x_i\in X$. Since those $3q$ vertices are neighbors of $\bigcup_{j\in J^+} \{ \down{j}{1}, \down{j}{2}, \down{j}{3}\} $ and $|J^+|=q$, the construction of $\Gdthree$ implies $ \bigcup_{j\in J^+ } C_j = X $. Hence every element of $X$ is contained in exactly one set of $\mathcal{C}'=\{ C_j\in \mathcal{C}: j\in J^+  \}$,
implying the instance $(\mathcal{C},X)$ is a YES-instance.
\end{proof}

Lemma \ref{lemmaifYESthennuisnud} and Lemma \ref{lemmaifnuisnudthenYES} show NP-hardness of deciding $\nu(G')=\nu_d(G')$ for a given bipartite graph $G'$ with maximum degree three. Furthermore, deciding if $\nu(G') = \nu_d(G')$ is clearly in NP, since a disconnected matching of size $\nu(G')$ is a certificate for a YES instance, completing the proof of Theorem \ref{theoremnunuddelta3}.\\

We finish this section with a corollary, which follows immediately from Theorem~\ref{theoremnunuddelta3}.
\begin{corollary}\label{discNPc}
The disconnected matching number is \NPH\text{} to compute for graphs with maximum degree three. The corresponding decision problem \textsc{Disconnected Matching}, i.e. deciding if $v_d\ge k$ for some integer $k$, is \NPc\text{} for graphs with maximum degree three.
\end{corollary}

\section{\texorpdfstring{$ \nu=\nu_d $ for diameter at most 3?}{Disconnected equals unrestricted for diameter at most 3}}

For our last result regarding the equality of the matching number and disconnected matching number, we consider graphs with diameter 3.

\begin{theorem}\label{theoremdiam3}
Given $j\ge 2$,
deciding if $\nu_{d,j}(G)=\nu(G)$ is in \P\text{} for graphs of diameter at most 3.
\end{theorem}
Theorem \ref{theoremdiam3} follows immediately from the characterization of diameter 3 graphs with equal matching number and disconnected matching number in the following lemma:

\begin{lemma}\label{lemmadiam3}
Given a graph $G$ with $m(G)\ge 1$ of diameter at most 3 and $j\ge 2$, $\nu(G)=\nu_{d,j}(G)$ if and only if there exists a vertex $v$ such that $G-v$ contains at least $j$ many non-trivial components and $\nu(G)=\nu(G-v)$.
\end{lemma}

\begin{proof}
The existence of such a vertex $v$ clearly implies $\nu(G)=\nu_{d,j}(G)$. Now we prove the other direction.\\
Let $\nu(G)=\nu_{d,j}(G)$ and let $M$ be a maximum $j$-disconnected matching of $G$, i.e. $|M|=\nu(G)$ and $G[M]$ has at least $j$ many components.
Note that $G-G[M]$ contains no edge, since otherwise we could add such an edge to $M$, implying $\nu(G)>\nu_{d,j}(G)$, contradiction.\\

\noindent\textbf{Claim 1:} If there exist saturated vertices without non-saturated neighbors, then they
are in the same component of $G[M]$.\\
Assume there exist saturated vertices $v,v'$ without non-saturated neighbors in different components of $G[M]$ and let $vu_1u_2\cdots u_\ell v'$ be a shortest path between $v$ and $v'$. It follows that $u_1,u_\ell$ are saturated vertices in the same components as $v,v'$, respectively. Since the path has to contain a non-saturated vertex as well, it follows that $v$ and $v'$ have distance at least 4, contradiction. \\

If vertices without non-saturated neighbors exist, let $C_0$ be the component 
of $G[M]$ that contains all those vertices
and let $C_1,\ldots, C_k$ be the remaining components. Hence every vertex in $C_i$, $i\ge 1$, has a non-saturated neighbor.
For every $i\ge 1$, let $w_i$ be a non-saturated vertex adjacent to at least one vertex in the component $C_i$.\\

\noindent\textbf{Claim 2:}
If $vv'\in M \cap E(C_i)$ for $i\ge 1$,
then $N(\{v,v'\})$ contains exactly one non-saturated vertex.\\
Note that both $v$ and $v'$ have a non-saturated neighbor. 
Assume there exists different non-saturated vertices $w$ and $w'$ 
such that $vw, v'w' \in E(G)$. 
Then $wvv'w'$ is an augmenting path, implying $\nu(G)>\nu_{d,j}(G)$, contradiction.\\

\noindent\textbf{Claim 3:} For $i\ge 1$, $w_i$ is adjacent to every vertex in $C_i$.\\
Assume there exists $i\ge 1$ such that Claim 3 is false. 
Let $U$ contain all vertices from $C_i$ that are adjacent to $w_i$.
There exists a vertex $v\in U$ with a neighbor $v'\in V(C_i)\setminus U$. Let $w'$ be a non-saturated neighbor of $v'$.
If $vv'\in M$, then $w_ivv'w'$ is an augmenting path, implying $\nu(G)>\nu_d(G)$, contradiction.
If $vv'\not \in M$, then exist $u,u'$ such that $vu,v'u'\in M$. By Claim 2, $uw_i,u'w' \in E(G)$. Hence $w_iuvv'u'w'$ is an augmenting path, implying $\nu(G)>\nu_{d,j}(G)$, contradiction.\\

Since for $i\ge 1$ the vertex $w_i$ is adjacent to every vertex in $C_i$, Claim 2 implies that $w_i$ is the unique non-saturated neighbor for every vertex in $C_i$. \\

\noindent\textbf{Claim 4:} $w_i=w_{k}$ for every $i,k\ge 1$.\\
Assume $w_i\not=w_k$.
Let $v_i\in C_i$ and $v_k\in C_k$. Every path between $v_i$ and $v_k$ contains $w_i$ and $w_k$. Since $w_i$ and $w_k$ are not adjacent, 
$v_i$ and $v_k$ have distance at least 4, contradiction.\\

Therefore $w_1$ is the unique non-saturated vertex adjacent to any and every vertex in $C_i$ for $i\ge 1$. If $C_0$ exists, every path between a vertex in $C_0$ and a vertex in $C_i$ for $i\ge 1$ contains the vertex $w_1$.
Hence the number of non-trivial components of $G-w_1$ is equal to the number of components of $G[M]$, but $\nu(G)=\nu(G-w_1)$ since $w_1$ is not saturated by $M$, completing the proof of Lemma \ref{lemmadiam3}.
\end{proof}

If $G$ has diameter 2, then two vertices from different components of $G[M]$ must have a common non-saturated neighbor, implying every saturated vertex has a non-saturated neighbor. Then $w_1$ is a universal vertex and $G-w_1$ has a perfect matching, implying the following corollary: 
\begin{corollary}\label{lemmadiam2}
Given a graph $G$ with diameter 2 and $j\ge 2$, $\nu_{d,j}(G)=\nu(G)$ if and only if there exists a universal vertex $v$ such that $G-v$ has a perfect matching and $G-v$ has at least $j$ many non-trivial components.\end{corollary}

\section{\texorpdfstring{$\nu_d = \nu_s$?}{Disconnected matching number equals induced}}

We now focus on the disconnected matching number and induced matching number. In this section, we show that deciding equality of the disconnected matching number and induced matching number for bipartite graphs is \coNPc.
To achieve this result, we use simple bounds on the disconnected matching number and induced matching number, captured in Lemma \ref{inequalitynusle1maxnud} and Lemma \ref{nuslemax1}.

\begin{theorem}\label{nudnuscoNPcomplete}Given a bipartite graph $G$ of diameter at most 3, deciding if $\nu_d(G) = \nu_s(G)$ is \coNPc.
\end{theorem}
In order to show Theorem \ref{nudnuscoNPcomplete}, we first proof the following lemma:
\begin{lemma}\label{inequalitynusle1maxnud}
For any graph $G$:
\begin{align*}
\nu_s(G) \le 1+ \max\limits_{uv\in E(G)} \nu \left(G-N\left[ \{u,v\} \right] \right) \le \nu_d(G)
\end{align*}
\end{lemma}

\begin{proof}Since for any edge $uv$ in $G$ a matching of $G-N[ \{ u,v \}  ]$ together with the edge $uv$ yields a disconnected matching, 
it holds that
\begin{align*}
\nu_d(G)\ge 1+\max \limits_{uv\in E(G)} \nu\left(G-N[\{ u,v \}]\right).
\end{align*}

Let $M$ be a maximum induced matching of $G$
and let $e\in M$. Note that $M\setminus \{e\}$ is an induced matching of $G-N[e]$ and obviously $|M\setminus \{ e \}|\le \nu\left(G-N[e] \right)$. Hence $|M|\le 1+ \nu\left(G-N[e] \right) \le  1+\max \limits_{uv\in E(G)} \nu\left(G-N[ \{ u,v  \} ] \right) $, implying 
\begin{align*}
\nu_s(G) \le 1+\max \limits_{uv\in E(G)} \nu\left(G-N[ \{ u,v  \} ] \right).
\end{align*}
This concludes the proof of Lemma \ref{inequalitynusle1maxnud}.
\end{proof}

To show that deciding if $\nu_d(G)=\nu_s(G)$ is
\coNPH\text{}
for bipartite graphs,
we give a polynomial-time reduction from the decision problem \pname{Disconnected Matching}
to the complement of $\nu_d=\nu_s$. 

An instance of \pname{Disconnected Matching} consists of a graph $G$ and an integer $k$. 
We have to decide if $G$ has a disconnected matching of size at least $k$, i.e. if $\nu_d(G)\ge k$. \pname{Disconnected Matching} is {\NPc} for bipartite graphs~\cite{disconnected_matchings_cocoon}.
Note that, by Lemma~\ref{inequalitynusle1maxnud}, we can assume that $k> 1 + \max_{uv\in E(G)} \nu\left( G-N[ \{u,v\} ] \right)$, since otherwise we could easily find a disconnected matching of size at least $k$ containing an edge $uv$ that maximizes $\nu\left(G-N[\{ u,v \}]\right)$ as well as a maximum matching of $G - N[{u,v}]$.\\

Now we describe our reduction. Let $G$ be a bipartite graph and let $k$ be an integer satisfying $k>  1 + \max_{uv\in E(G)} \nu\left( G-N[ \{u,v\} ] \right)$. Let $A\mathbin{\dot{\cup}} B$ be the partition of $V(G)$.

We construct $G'$ by adding $2(k-1)$ many vertices $u_1,\ldots,u_{k-1},v_1,\ldots,v_{k-1}$ and $k-1$ edges, $e_1=u_1v_1,\ldots, e_{k-1}=u_{k-1}v_{k-1}$, to $G$. We add all possible edges between $\{ v_1,\ldots, v_{k-1} \}$ and $A$ and between $\{u_1,\ldots, u_{k-1}\}$ and $B$ 
(see Figure \ref{gnewnudnus}).
The new graph $G'$ is bipartite and has diameter at most 3.

\begin{figure}[htb!]
	\centering
	\begin{tikzpicture}
	
	\tikzstyle{convcols}=[color=black]
	\tikzstyle{point}=[draw,circle,inner sep=0.cm, minimum size=1.5mm, fill=black]
	\tikzstyle{point2}=[draw,circle,inner sep=0.cm, minimum size=0.5mm, fill=black]

	\node[point] at (0,0) [label=left:$u_{k-1}$] {};
	\node[point] at (0,2.5) [label=left:$u_2$] {};
	\node[point] at (0,4.5) [label=left:$u_1$] {};

	\node[point] at (0,1) [label=left:$v_{k-1}$] {};
	\node[point] at (0,3.5) [label=left:$v_2$] {};
	\node[point] at (0,5.5) [label=left:$v_1$] {};

	\foreach \i in {0,2.5,4.5}
	{	
		\draw (0,\i) -- (0,\i+1);
	}
		
	\node[point2] at (4,2) [label=above:] {};
	\node[point2] at (4,1.75) [label=above:] {};
	\node[point2] at (4,1.5) [label=above:] {};
		
	\node[point2] at (0,1.5) [label=above:] {};
	\node[point2] at (0,1.75) [label=above:] {};
	\node[point2] at (0,2) [label=above:] {};

  \node[point2] at (5,2.25) [label=above:] {};
	\node[point2] at (5,2.5) [label=above:] {};
	\node[point2] at (5,2) [label=above:] {};
	
	\foreach \j in {0,3,4.5}{
    	\node[point] at (4,\j) [label=above:] {};
 
		\foreach \i in {0,2.5,4.5}
			{
		\draw (0,\i) -- (4,\j);	
		}
	}

 	\foreach \j in {0,4.5}{
    	\node[point] at (5,\j) [label=above:] {};
 
		\foreach \i in {1,3.5,5.5}
			{
		\draw (0,\i) -- (5,\j);	
		}
	}

	\draw[convcols] (3.5,-0.5) rectangle ++(1,5.5) node[below]{};
	\draw[convcols] (4.5,-0.5) rectangle ++(1,5.5) node[midway]{};
	\node at (4,-1) { $A$};
	\node at (5,-1) { $B$};
		
	\draw[thick,black,decorate,decoration={brace,amplitude=10}] (3.3,5.5) -- (5.7,5.5) node[midway, above,yshift=10]{$G$};	
	
	\end{tikzpicture}
	\caption{$G'$}
	\label{gnewnudnus}
\end{figure}

The following lemma implies \coNPHness\text{} of the decision problem $\nu_d= \nu_s$. 

\begin{lemma}\label{reductionGprime}
$\nu_s(G')\not=\nu_d(G')$ if and only if $\nu_d(G)\ge k$.
\end{lemma}

\begin{proof}
By construction of $G'$, any disconnected (or induced) matching of $G'$ is either a subset of $\{e_1,\ldots, e_{k-1}\}$ or $E(G)$, implying $\nu_d(G')=\max \{k-1, \nu_d(G) \}$ and $\nu_s(G')=\max \{k-1, \nu_s(G) \}$.\\
\newline
($\Rightarrow$): 
Let $\nu_d(G)\ge k$. Since $k> 1 + \max_{uv\in E(G)} \nu\left( G-N[ \{u,v\} ] \right)$, Lemma~\ref{inequalitynusle1maxnud} implies $k> 
\nu_s(G)$. Hence $\nu_s(G')=\max\{ k-1, \nu_s(G)  \}=k-1$, implying $\nu_d(G')\ge \nu_d(G)\ge k>\nu_s(G')$.\\
\newline
($\Leftarrow$): Let $\nu_s(G')\not=\nu_d(G') $.
$\{e_1,\ldots, e_{k-1}\}$ is an induced matching of $G'$, implying $\nu_s(G')\ge k-1$. Since $\nu_s(G')\not=\nu_d(G') $, we get $k \le \nu_d(G')=\max\{ k-1, \nu_d(G)  \}$ and hence $ \nu_d(G) \ge k$.
\end{proof}

From Lemma~\ref{inequalitynusle1maxnud} we deduce the following corollary.

\begin{corollary}
	For any graph $G$, if $\nu_d(G)=\nu_s(G)$ then
	\begin{align*}
	\nu_d(G)=\nu_s(G)=1+\max \limits_{uv\in E(G)} \nu\left(G-N[\{ u,v \}]\right)
	\end{align*}
	and we can find an induced matching of this size in polynomial time.\\
	If $\nu_d(G)=\nu_s(G)$ and we choose an edge $uv$ that belongs to a maximum induced matching, the remaining graph $G-N[\{u,v\}]$ has equal matching number and induced matching number and on those graphs a maximum induced matching can be found in polynomial time \cite{koblerrotics2003, cameronwalker2005}.\end{corollary}

It remains to show that 
deciding if $\nu_d(G)=\nu_s(G)$ is in \coNP. We use a slightly stronger upper bound for the induced matching number.

\begin{lemma}\label{nuslemax1}
For any graph $G$, it holds $\nu_s(G)\le s(G)$ for \begin{align*}
s(G)\coloneqq \max\limits_{uv\in E(G)} \bigg( \nu\left(G- N[ \{u,v\} ]\right) +\mathds{1}\Big[\nu\left(G-N[ \{u,v\} ]\right)=\nu_s\left(G-N[ \{u,v\} ]\right) \Big] \bigg),
\end{align*}
where $\mathds{1}\big[ A \big]$ is a function receiving a statement $A$ and returning $1$ if the statement $A$ is true and $0$ otherwise.
		\end{lemma}

\begin{proof}
Let $M$ be a maximum induced matching of $G$
and $e\in M$. 
Clearly, $M\setminus \{e\}$ is an induced matching of $G-N[e]$. We consider two cases:
\begin{itemize}
\item If $\nu\left(G-N[ e ]\right)=\nu_s\left(G-N[e ]\right)$, then $|M\setminus \{ e \}|= \nu\left(G-N[e] \right)$ and hence $|M|= 1+ \nu\left(G-N[e] \right)$.
\item If $\nu\left(G-N[e ]\right)\not=\nu_s\left(G-N[e]\right)$, then $|M\setminus \{ e \}|< \nu\left(G-N[e] \right)$ and hence $|M|\le \nu\left(G-N[e] \right)$.
\end{itemize}
Combining those two cases yields \begin{align*}
\nu_s(G)=|M|\le \nu\left(G-N[e] \right) + \mathds{1}\Big[\nu\left(G-N[e ]\right)=\nu_s\left(G-N[ e ]\right) \Big] \le s(G),
\end{align*}
concluding the proof.
\end{proof}

Note that deciding if the matching number and induced matching number of a given graph are equal can be done in polynomial time \cite{koblerrotics2003}.
Hence for any graph $G$ the upper bound $s(G)$ in Lemma~\ref{nuslemax1} can be computed in polynomial time.

Now we use $s(G)$ to characterize graphs with equal disconnected matching number and disconnected matching number.

\begin{lemma}\label{lemmabound}
	For any graph $G$:
	\begin{align*}
	\nu_d(G)\not= \nu_s(G) \text{ if and only if }\nu_d(G)> s(G)
	\end{align*}
\end{lemma}

\begin{proof}$ $\\
($\Leftarrow$): If $\nu_d(G)>s(G)$, then $\nu_d(G)>s(G)\ge \nu_s(G)$ by Lemma~\ref{nuslemax1}.	\\
\newline
($\Rightarrow$): If $\nu_d(G)\not=\nu_s(G)$, we consider two cases:

\begin{enumerate}
\item If $s(G)= \max\limits_{uv\in E(G)} \nu\left(G- N[ \{u,v\} ]\right) +1$, there exists an edge $u'v'$ such that
\begin{align*}
&\nu\left(G- N[ \{u',v'\} ]\right) +\mathds{1}\Big[\nu\left(G-N[ \{u',v'\} ]\right)=\nu_s\left(G-N[ \{u',v'\} ]\right) \Big] \\
&= \max\limits_{uv\in E(G)} \nu\left(G- N[ \{u,v\} ]\right) +1.
\end{align*}
Note that $\nu\left(G-N[ \{u',v'\} ]\right)=\nu_s\left(G-N[ \{u',v'\} ]\right)$. The edge $u'v'$ together with a maximum induced matching of $G-N[ \{u',v'\} ]$ yields an induced matching of size $s(G)$.
Lemma~\ref{nuslemax1} then implies $\nu_s(G)=s(G)$. By assumption $\nu_d(G)\not=\nu_s(G)$ and hence $\nu_d(G)>\nu_s(G)= s(G)$.\\

\item If $s(G)< \max\limits_{uv\in E(G)} \nu\left(G- N[ \{u,v\} ]\right) +1$, then Lemma~\ref{inequalitynusle1maxnud} implies:
\begin{align*}\nu_d(G)\ge \max\limits_{uv\in E(G)} \nu\left(G- N[ \{u,v\} ]\right) +1 >s(G)
\end{align*}
\end{enumerate}	
This concludes the proof of Lemma \ref{lemmabound}.
\end{proof}

By Lemma~\ref{lemmabound}, verifying that a given graph $G$ is a \NOi-instance of the decision problem $\nu_d=\nu_s$ can be done efficiently by giving a disconnected matching of size larger than $s(G)$ as a certificate. Hence the decision problem $\nu_d=\nu_s$ is indeed in \coNP. 

\section{\texorpdfstring{$\nu_d = \nu_s$ for bounded degree?}{Disconnected equals induced for bounded degree}}

Both the disconnected matching number (see Corollary \ref{discNPc}) and induced matching number are \NPH\text{} to compute on graphs with bounded degree \cite{np_completeness_of_some_generalizations_of_the_maximum_matching_problem}.
However, we show that one can decide in polynomial time if those two numbers are equal in graphs with bounded degree, captured in the following theorem. 
\begin{theorem}\label{theorempoly} Deciding if the disconnected matching number and induced matching number are equal 
can be done in polynomial time
	for graphs with bounded degree.
\end{theorem}

To prove Theorem \ref{theorempoly},
we characterize graphs with equal induced matching number and disconnected matching number when the induced matching number is large enough compared to the maximum degree (see Theorem \ref{theoremstructure}). 
Our result is similar to the result from Cameron and Walker \cite{cameronwalker2005}, where they showed, that a connected graph has identical matching number and induced matching number, if there exists sets $A,B,C,D$ as in Definition \ref{defcameronwalkergraph}.

\begin{definition}\label{defcameronwalkergraph}A connected graph $G$ is a Cameron-Walker-graph, if $V(G)$ is the disjoint union of four sets $A,B,C,D$ where 
\begin{itemize}
\item all vertices in $A$ have degree 1,    \item if $G[A\cup B]$ is non-empty, it is 1-regular and bipartite with partition $A\mathbin{\dot{\cup}}B$,
    \item $G[B\cup C]$ is bipartite with partition $B\mathbin{\dot{\cup}}C$,
    \item if $G[D]$ is non-empty, it is 1-regular and for every edge $xy\in G[D]$ exists $z\in C$ s.t. $N_G(x)\setminus \{y\}=N_G(y)\setminus \{x\}=\{z\}$ (see Figure \ref{decomp}).
           \end{itemize}
\end{definition}
\begin{figure}[htb!]
	\centering
	\begin{tikzpicture}[rotate=90, xscale=0.5]
	\tikzstyle{point}=[draw,circle,inner sep=0.cm, minimum size=1mm, fill=black]
	\tikzstyle{point2}=[draw,circle,inner sep=0.cm, minimum size=0.333mm, fill=black]

	\foreach \j in {3,4}{
		\node[point] (u1) at (0.5,\j-0.25) [label=below:] {};
		\node[point] (u2) at (0.5,\j+0.25) [label=above:] {};
		\node[point] (v1) at (2,\j) [label=below:] {};
		\draw (u1) -- (u2);
		\draw (u1) -- (v1);
		\draw (u2) -- (v1);
	}

	\node[point] at (2,1) [label=below:] {};
	\node[point] at (2,2) [label=below:] {};
	\node[point2] at (2,0.25) [label=above:] {};
	\node[point2] at (2,0) [label=above:] {};
	\node[point2] at (2, -0.25) [label=above:] {};

	\node[point2] at (0.5,1.75) [label=above:] {};
	\node[point2] at (0.5,2) [label=above:] {};
	\node[point2] at (0.5, 2.25) [label=above:] {};
	
	\draw (1.5,-0.5) rectangle ++(1,5) node[below]{};
	\node at (2,-1) { $C$};
	
	\draw (0,-0.5) rectangle ++(1,5) node[below]{};
	\node at (0.5,-1) { $D$};
	
	\begin{scope}[shift={(-1,0)}]
	\foreach \j in {0,2,3,4}{		
		\node[point] (u1) at (4,\j) [label=below:] {};
		\node[point] (u2) at (5.5,\j) [label=above:] {};
		\draw (u1) -- (u2);
	}
	
	\node[point2] at (4,0.75) [label=above:] {};
	\node[point2] at (4,1) [label=above:] {};
	\node[point2] at (4, 1.25) [label=above:] {};
	\node[point]  (a4) at (4,0) [label=above:] {};	
	
	\draw (3.5,-0.5) rectangle ++(1,5) node[below]{};
	\node at (4,-1) { $B$};
	
	\begin{scope}[shift={(-0.5,0)}]
	
	\draw (5.5,-0.5) rectangle ++(1,5) node[midway]{};
	\node[point2] at (6,0.75) [label=above:] {};
	\node[point2] at (6,1) [label=above:] {};
	\node[point2] at (6, 1.25) [label=above:] {};
	\node at (6,-1) { $A$};
	\end{scope}
	
	\end{scope}

	\end{tikzpicture}
	\caption{Cameron-Walker-graph}
	\label{decomp}
\end{figure}
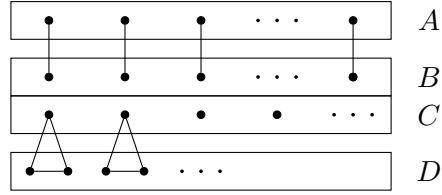
Cameron-Walker-graphs can be recognized in polynomial time by searching for triangles and vertices of degree $1$. 
If $G$ is disconnected and contains at least two non-trivial components, then $\nu_{d}(G)=\nu(G)$ and hence $\nu_d(G)=\nu_s(G)$ if and only if every non-trivial component of $G$ is a Cameron-Walker-graph. 
Therefore, it remains to show Theorem \ref{theorempoly} for connected graphs.
We introduce the following characterization.

\begin{theorem}\label{theoremstructure} If $G$ is a connected graph with $\nu_s(G)\ge 2 \Delta(G)$, then $\nu_s(G)=\nu_{d}(G)$ if and only if $G$ is a Cameron-Walker-graph.
\end{theorem}

From Theorem \ref{theoremstructure}, we deduce the following polynomial time algorithm.
\begin{proof}[Proof of Theorem \ref{theorempoly}]
Let $G$ be a graph with bounded maximum degree.
We check by brute force, if $G$ has an induced matching of size $2\Delta(G)$. If yes, according to Lemma \ref{theoremstructure} it is enough to test if $G$ is a Cameron-Walker-graph, which can be done in polynomial time.
If no, we calculate the induced matching number and disconnected matching number by brute force and compare them.
\end{proof}

It remains to proof Theorem \ref{theoremstructure}. 
We first introduce and show the following two lemmas -- Lemma \ref{lemmanoremainingedges} and Lemma \ref{lemmatriangleordegree1}.

\begin{lemma}\label{lemmanoremainingedges}
	If $G$ is a graph with $\nu_d(G)=\nu_s(G)>2 \Delta(G)-2 $ and $M$ is a maximum induced matching in $G$, then $G-G[M]$ contains no edge.
\end{lemma}
\begin{proof}
	Let $M$ be a maximum induced matching. Assume $G-G[M]$ contains an edge $uv$. Since $|N_G(\{u,v\})\setminus \{u,v\}|\le 2\Delta(G)-2$ and $|M|>2\Delta(G)-2$, adding $uv$ to $M$ yields a strictly larger disconnected matching, contradicting  $\nu_d(G)=\nu_s(G)$.
\end{proof}	

Kobler and Rotics \cite{koblerrotics2003} argued, that in a graph with equal matching number and induced matching number, each edge of a maximum induced matching either contains a vertex of degree one or the edge itself is contained in a triangle. The same is true for a maximum induced matching in a graph $G$ with $\nu_s(G)\ge 2 \Delta(G)$.

\begin{lemma}\label{lemmatriangleordegree1}
	Let $G$ be a graph with $\nu_s(G)\ge 2\Delta(G)$ and $\nu_s(G)=\nu_{d}(G)\ge 2$. Let $M$ be a maximum induced matching in $G$. For any edge $xy\in M$, there are two possibilities:
	\begin{enumerate}
		\item $xy$ lies in a triangle and $d_G(x)=d_G(y)=2$ or
		\item $d_G(x)=1$ or $d_G(y)=1$
	\end{enumerate}
\end{lemma}

\begin{proof}Let $xy\in M$. We distinguish the cases that $xy$ lies in a triangle or not.
	\begin{enumerate}
		\item Let $xy$ lie in a triangle $xyz$. 
     		Assume that $x$ has degree at least $3$ and let $u\in N_G(x)\setminus \{y,z\}$. Note that both $u$ and $z$ are not saturated by $M$.
      Hence removing the edge $xy$ from the matching $M$ and adding the edges $xu$ and $yz$ yields a new strictly larger matching, say $M'$.
    To receive a contradiction, we show that $M'$ is disconnected. Since $M$ is induced, it is enough to show that not every edge in $M\setminus \{xy\}$ intersects the neighborhood of $z$ or $u$.
    	\begin{align*}
		&\left|N_G(\{z,u\})\setminus \{ x,y \}\right| \\
		&\le \left|N_G(z)\setminus \{x,y\}\right| +  \left| N_G(u)\setminus \{x\} \right| \\
		&\le \Delta(G)-2+\Delta(G)-1\\
		&=2\Delta(G)-3
		\end{align*}
But 
$|M\setminus \{xy\}|=\nu_s(G)-1>2\Delta(G)-3$, implying $M'$ is disconnected, contradiction. 
 Therefore $x$ and $y$ have degree $2$.

		\item Let $xy \in M$ be not contained in a triangle of $G$.
		Assume $d_G(x)\ge 2$ and $d_G(y)\ge 2$. Then there exists $z\in N_G(x)\setminus \{y\}$ and $z'\in N_G(y)\setminus\{x\}$. $z\not=z'$, otherwise $xyz$ would be a triangle containing the edge $xy$. Since $M$ is induced, neither $z$ nor $z'$ is saturated by $M$. Removing the edge $xy$ and adding the edges $xz$ and $yz'$ yields a strictly larger matching $M'$.
        Analogously, $\left|N_G(\{z,z'\})\setminus \{ x,y \}\right|\le 2\Delta(G)-2$ and 
                $|M\setminus \{xy\}|>2\Delta(G)-2$ imply that the matching $M'$ is
                is disconnected, contradiction. Hence $d_G(x)=1$ or $d_G(y)=1$.
	\end{enumerate}
	This completes the proof of Lemma \ref{lemmatriangleordegree1}.
\end{proof}

\begin{proof}[Proof of Theorem~\ref{theoremstructure}] Let $G$ be a connected graph with $\nu_s(G)\ge 2\Delta(G)$.\\
	$(\Rightarrow)$ Let $\nu_d(G)=\nu_s(G)$ and let $M$ be a maximum induced matching in $G$.
 If $n(G)=1$, 
the existence of such a decomposition is clear. $G$ can't be isomorphic to a $P_2$. 
So let $\Delta(G)\ge 2$ and hence $\nu_s(G)\ge 4$.\\
	Let $A$ be the set of all vertices of degree $1$ saturated by $M$. Let $D$ be all vertices saturated by $M$ contained in a triangle. Choose $B=N_G(A)$ and hence $E_G(A,B)\subseteq M$. $B$ is independent since $M$ is induced. By Lemma~\ref{lemmatriangleordegree1}, $A\cap D=\emptyset$. Let $C=V(G)\setminus (A\cup B \cup D)$ be all remaining vertices. By Lemma~\ref{lemmatriangleordegree1} any edge in $M$ is either an edge of a triangle or has one vertex of degree $1$, implying $M= E_G(A,B) \cup E(G[D])$. 
	Since $G-G[M]$ contains no edge by Lemma~\ref{lemmanoremainingedges}, $C$ is an independent set.\\
	By choice of $D$, any edge $xy$ in $G[D]$ lies in a triangle $xyz$. Since $z$ is not saturated by the induced matching $M$, $z$ lies in $C$. By Lemma~\ref{lemmatriangleordegree1} $d_G(x)=d_G(y)=2$, hence $\{z\}= N_G(x)\setminus \{y\}=N_G(y)\setminus \{x \}$ and $G[D]$ is 1-regular.\\
	\newline
	$(\Leftarrow)$ Let $V(G)$ be the disjoint union of four sets $A,B,C,D$ with the corresponding properties. $G-D$ is bipartite with partition $(A\cup C) \mathbin{\dot{\cup}}B$ and hence $\nu(G-D)\le |B|=|A|$. Since any matching can only contain one edge of each triangle, it follows that $\nu_s(G)\le \nu_d(G)\le \nu(G)\le \nu(G-D)+\frac{|D|}{2}=|A|+\frac{|D|}{2}$. On the other hand, $E_G(A,B)\cup E(G[D])$ is an induced matching of size $|A|+\frac{|D|}{2}$, implying $\nu_s(G)=\nu_d(G)$.
	\end{proof}

\section{Sequence of disconnected matching numbers}

One question that may arise is if there is any relation between disconnected matching numbers for general graphs. For instance, given a graph $G$ and two of its disconnected matching numbers $\nu_{d,i}$ and $\nu_{d,i+2}$, 
what do we know about $\nu_{d,i+1}$?
More generally, given finitely many disconnected matching numbers $\nu_{d,i}(G)$, $i\in I\subseteq \mathbb{N}$, can we approximate the value of $\nu_{d,j}(G)$ for some $j\not \in I$? We answer this question negatively. The following theorem implies, that, without further restrictions on the graph, we get only the trivial bounds
$\nu_{d,i}(G) \ge \nu_{d,i+1}(G)$ and, if $G$ has a non-empty $i$-disconnected matching, $\nu_{d,i}(G)\ge i$.

\begin{theorem}\label{theoremmainsequence}
Given a finite non-increasing sequence of integers $\beta_1\ge \beta_2 \ge \cdots \ge \beta_{k-1} \ge \beta_k$ with $\beta_k\ge k$, there exists a graph $G$ with $\nu_{d,i}(G)=\beta_i$ for all $i\in [k]$.
\end{theorem}

In order to show Theorem \ref{theoremmainsequence},
we recursively construct a sequence of graphs $G_k,G_{k-1},G_{k-2},\ldots, G_1$, where $G=G_1$ is a valid choice for the graph in Theorem~\ref{theoremmainsequence}.\\
Let $G_k$ contain $k -1$ many $P_2$-components (say $E_1,\ldots, E_{k-1}$) and a clique of size $2 (\beta_k-(k-1))$, yielding a graph with $k$ many components and a perfect matching of size $\beta_k$ (see Figure \ref{Gksequence}).

Let $1\le i<k$ and $G_{i+1}$ be given.
\begin{itemize}
\item We construct $G_i$ by adding a clique of size $2(\beta_i-\beta_{i+1})$ to $G_{i+1}$ and we add all possible edges between the clique and $G_{i+1}- E_1\cup \cdots \cup E_{i-1}$ (see Figure \ref{Gisequence}).\\
Note that if $\beta_{i}=\beta_{i+1}$, then $G_{i+1} = G_i$. 
\end{itemize}

\begin{figure}[htb!]
	\centering
	\begin{tikzpicture}
	\tikzstyle{point}=[draw,circle,inner sep=0.cm, minimum size=1.5mm, fill=black]
	\tikzstyle{point2}=[draw,circle,inner sep=0.cm, minimum size=0.5mm, fill=black]
	
	\foreach \i in {1,2,3}
	{	
		\node[point] at (-\i,0) [label=below:$E_\i$] {};
		\node[point] at (-\i,1) [label=below:] {};
		\draw (-\i,0) -- (-\i,1);
	}
	
	\node[point2] at (-4,0.5) [label=above:] {};
	\node[point2] at (-4.25,0.5) [label=above:] {};
	\node[point2] at (-3.75,0.5) [label=above:] {};
	
	\node[point] at (-5,0) [label=below:$E_{k-1}$] {};
	\node[point] at (-5,1) [label=below:] {};
	\draw (-5,0) -- (-5,1);
	
	\draw (-10,0) rectangle ++(4,1) node[midway]{$K_{2(\beta_k-(k-1))}$};
	\end{tikzpicture}
	\caption{$G_k$}
	\label{Gksequence}
\end{figure}

\begin{figure}[htb!]
	\centering
	\begin{tikzpicture}
	
	\tikzstyle{point}=[draw,circle,inner sep=0.cm, minimum size=1.5mm, fill=black]
	\tikzstyle{point2}=[draw,circle,inner sep=0.cm, minimum size=0.5mm, fill=black]
		
\node[point] at (1,0) [label=below:$E_1$] {};
\node[point] at (1,1) [label=below:] {};
\draw (1,0) -- (1,1);

\node[point2] at (-0,0.5) [label=above:] {};
\node[point2] at (-0.25,0.5) [label=above:] {};
\node[point2] at (0.25,0.5) [label=above:] {};

\node[point] at (-1,0) [label=below:$E_{i-1}$] {};
\node[point] at (-1,1) [label=below:] {};
\draw (-1,0) -- (-1,1);

\node[point] at (-2,0) [label=below:$E_{i}$] {};
\node[point] at (-2,1) [label=below:] {};
\draw (-2,0) -- (-2,1);

\node[point] at (-3,0) [label=below:$E_{i+1}$] {};
\node[point] at (-3,1) [label=below:] {};
\draw (-3,0) -- (-3,1);

\node[point2] at (-4,0.5) [label=above:] {};
\node[point2] at (-4.25,0.5) [label=above:] {};
\node[point2] at (-3.75,0.5) [label=above:] {};

\node[point] at (-5,0) [label=below:$E_{k-1}$] {};
\node[point] at (-5,1) [label=below:] {};
\draw (-5,0) -- (-5,1);

	\draw (-8.5,-0.6) rectangle ++(6,2.2) ;
	\node at (-7,0.5) [label=left:] {};
	
	\draw (-11,-0.5) rectangle ++(2,2) node[midway]{$K_{2(\beta_i-\beta_{i+1})}$};
	
	\draw (-11.5,-0.8) rectangle ++(10,2.6);
	\node at (-6,-1) [label=below:] {};
		
	\draw (-8,-0.4) rectangle ++(2.2,1.8) node[midway]{$K_{2(\beta_{i+1}-\beta_{i+2})}$};
	
	\end{tikzpicture}
	\caption{$G_i$}
	\label{Gisequence}
\end{figure}

Inductively follows that $G_i$ contains a perfect matching of size $\beta_i$ for all $i\in [k]$. \\
For $1\le i< k$, if $\beta_i\not=\beta_{i+1}$, $G_i$ has $i$ many components, namely $E_1, \ldots, E_{i-1}$ and $G_i-E_1\cup \cdots \cup E_{i-1}$. If $\beta_i=\beta_{i+1}$, $G_i$ has at least $i+1$ many components and at least $i$ many $P_2$-components.

\begin{lemma}\label{lemmanuseq}
Let $i\in [k]$. $\nu_{d,i'}(G_i)=\beta_{i'}$ for all $i'\in \{i,i+1,\ldots,k\}$.
\end{lemma}

\begin{proof}We give a proof by induction.\\
$G_k$ has $k$ components and a perfect matching of size $\beta_k$, implying $\nu_{d,k}(G_k)=\beta_{k}$. Now let $1\le i <k$.
\begin{itemize}
\item If $\beta_i=\beta_{i+1}$, then $G_i=G_{i+1}$ and hence $\nu_{d,i'}(G_i)=\nu_{d,i'}(G_{i+1})=\beta_{i'}$ for $i'\in \{i+1,\ldots,k\}$ by induction hypothesis. $G_i$ has a perfect matching of size $\beta_{i}$ with at least $i+1$ many components, implying $\nu_{d,i}(G_i)=\beta_i$.
\item Let $\beta_i>\beta_{i+1}$.\\
Let $M$ be a matching in $G_{i}$ saturating at least one vertex $v$ of $G_{i}- G_{i+1} $. By construction, there are exactly $i-1$ many edges with no end-vertex in $N[v]$ (namely the edges in $E_1,\ldots, E_{i-1}$) and hence $M$ has at most $i$ many components. This implies that any matching in $G_i$ with at least $i+1$ many components cannot saturate any vertex in $G_{i}- G_{i+1}$. Hence $\nu_{d,i'}(G_i)=\nu_{d,i'}(G_{i+1})=\beta_{i'}$ for $i'\in \{ i+1,i+2,\ldots, k   \}$ by induction hypothesis.\\
By construction, the graph $G_i$ has $i$ many components and a perfect matching of size $\beta_i$, implying $\nu_{d,i}(G_i)=\beta_i$.
\end{itemize}
\vspace{-9mm}
\end{proof}

Lemma \ref{lemmanuseq} implies $\nu_{d,i}(G_1)=\beta_i$ for all $i\in [k]$, completing the proof of
Theorem~\ref{theoremmainsequence}.

\section*{Declarations} Partial financial support was received from research agencies CAPES, CNPq, FAPEMIG, and FAPERJ. The authors have no conflicts of interest to declare that are relevant to the content of this article.

\bibliographystyle{plain}
\bibliography{refs}

\end{document}